\documentclass{article}

\usepackage[symbol]{footmisc}
\usepackage[colorlinks=true, allcolors=blue]{hyperref}
\usepackage[utf8]{inputenc}
\usepackage[normalem]{ulem}
\usepackage{fullpage}
\usepackage{authblk}
\usepackage{enumitem}
\usepackage{xcolor}
\usepackage{diagbox}
\usepackage{makecell}
\usepackage{footnote}
\makesavenoteenv{tabular}
\makesavenoteenv{table}

\usepackage{tikz-cd, adjustbox}

\usepackage{amsmath, amsfonts, amsthm, amssymb, dsfont}
\usepackage{thmtools, thm-restate}
\usepackage{braket}

\declaretheorem{theorem}
\declaretheorem{lemma}

\declaretheorem{corollary}
\declaretheorem{observation}
\theoremstyle{definition}
\newtheorem{definition}{Definition}
\newtheorem*{definition*}{Definition}
\newtheorem*{remark}{Remark}

\begin{document}

\title{Exponential separations between classical and quantum learners}

\author[1]{Casper Gyurik 
\thanks{\href{mailto:c.f.s.gyurik@liacs.leidenuniv.nl}{c.f.s.gyurik@liacs.leidenuniv.nl}}}
\author[1]{Vedran Dunjko
\thanks{\href{mailto:v.dunjko@liacs.leidenuniv.nl}{v.dunjko@liacs.leidenuniv.nl}}}
\affil[1]{{\small applied Quantum algorithms (aQa), Leiden University, The Netherlands}} 

\date{\today}

\maketitle

\begin{abstract}
Despite significant effort, the quantum machine learning community has only demonstrated quantum learning advantages for artificial cryptography-inspired datasets when dealing with classical data. 
In this paper we address the challenge of finding learning problems where quantum learning algorithms can achieve a provable exponential speedup over classical learning algorithms. 
We reflect on computational learning theory concepts related to this question and discuss how subtle differences in definitions can result in significantly different requirements and tasks for the learner to meet and solve. 
We examine existing learning problems with provable quantum speedups and find that they largely rely on the classical hardness of evaluating the function that generates the data, rather than identifying it. 
To address this, we present two new learning separations where the classical difficulty primarily lies in identifying the function generating the data. Furthermore, we explore computational hardness assumptions that can be leveraged to prove quantum speedups in scenarios where data is quantum-generated, which implies likely quantum advantages in a plethora of more natural settings (e.g., in condensed matter and high energy physics).
We also discuss the limitations of the classical shadow paradigm in the context of learning separations, and how physically-motivated settings such as characterizing phases of matter and Hamiltonian learning fit  in the computational learning framework.
\end{abstract}

\section{Introduction}
\label{sec:introduction}

Quantum machine learning (QML)~\cite{biamonte:qml, arunachalam:qlearning} is a bustling field with the potential to deliver quantum enhancements for practically relevant problems.
An important goal of the community is to find practically relevant learning problems for which one can prove that quantum learners have an exponential advantage over classical learners.
In this paper, we study how to achieve such exponential separations between classical and quantum learners for problems with classical data in the efficient probably approximately correct (PAC) learning framework.
The results of this paper supersede the results of~\cite{gyurik:old}.
The first thing we address is that there is no single definition of what precisely constitutes a \textit{learning} separation.
In particular, when trying to come up with a definition there are many choices to be made, and various choices make sense depending on the particular settings.
For instance, as we explain, a significant difference arise if the emphasis is on the task of \textit{identifying} or \textit{evaluating} the functions that are generating the data.
This ambiguity can lead to conflating the task of learning in an intuitive sense with a purely computational task. 
To address this issue, we provide multiple definitions of a learning separation, and we discuss in which cases the tasks involve learning in an intuitive sense.
Moreover, we study existing learning separations~\cite{liu:dlp, servedio:q_c_learnability} and carefully delineate where the classical hardness of learning lies and the types of learning separations they achieve. 
Furthermore, we provide new examples of learning separations where the classical hardness lies more in learning in an intuitive sense rather than evaluating the functions to be learned.

Next, we turn our attention to the folklore in the community that states that quantum machine learning is most likely to have advantages when the data is quantum-generated.
For instance, it is believed that quantum learners are more likely to offer an advantage in predicting the phases of physical systems rather than distinguishing between images of dogs and cats. This is because genuine quantum-generated data typically has some $\mathsf{BQP}$-hard function underlying it. 
However, it is not immediately clear how these $\mathsf{BQP}$-hard functions can give rise to a learning separation. 
In other words, if we assume a complexity-theoretic separation like $\mathsf{BPP} \neq \mathsf{BQP}$, how can we construct a learning separation from the fact that the labeling function is $\mathsf{BQP}$-hard? 
In this paper, we address this question by exploring the additional complexity-theoretic assumptions required to build such a learning separation. 
Moreover, we provide several examples of how learning separations can be constructed from physical systems, such as the Bose-Hubbard model~\cite{childs:qma}, the antiferromagnetic Heisenberg and antiferromagnetic XY model~\cite{piddock:qma}, the Fermi-Hubbard model~\cite{gorman:electronic}, supersymmetric systems~\cite{cade:susy}, interacting bosons~\cite{wei:qma}, interacting fermions~\cite{liu:qma}.

\subsection*{Contributions}

The main contributions of this paper are as follows, listed according to the relevant sections:

\medskip

\noindent \textbf{\underline{Section 2}}:
\begin{itemize}
\item We clarify the finer points regarding the possible definitions of a learning separation by highlighting that there are various ways of defining them.
Above all, we explore the distinction between the task of identifying (i.e., giving a specification of) the correct labeling function versus the task of evaluating it (i.e., computing its output), and we explain that these differences have a significant impact on whether a problem exhibits a learning separation. 

\item We outline computational hardness assumptions that one can leverage to establish learning separations in the efficient PAC learning framework.
In particular, we define the complexity class $\mathsf{HeurBPP/samp}$ that aims to capture all classically learnable functions (see Definition~\ref{def:heurbpp/samp}).
Moreover, using ideas from~\cite{huang:power} we relate our new complexity class to a familiar though unexplored complexity class by showing that $\mathsf{HeurBPP/samp} \subseteq \mathsf{HeurP/poly}$ (see Lemma~\ref{lemma:samppoly}).
\end{itemize}

\noindent \textbf{\underline{Section 3}}:
\begin{itemize}
\item We discuss known learning separations~\cite{liu:dlp, servedio:q_c_learnability}, and we provide a fine-grained analysis of where the classical hardness of learning stems from.
\begin{itemize}
    \item We identify discrepancies between the definitions we posited in Section~\ref{sec:background} and the definitions of the learning separations of~\cite{liu:dlp, servedio:q_c_learnability}. Upon noticing these discrepancies, we find that it is not directly clear how the learning separations of~\cite{liu:dlp, servedio:q_c_learnability} follow from the canonical hardness assumptions of the related computational problems (i.e., the discrete logarithm and discrete cube root outlined in~\cite{blum:hardness} and~\cite{kearns:clt} respectively) when trying to make them comply with our definitions in Section~\ref{sec:background}.
    We address this as follows: 
    \begin{itemize}
    \item[(i)] We identify new stronger hardness assumptions for the relevant computational problems that adequately support a proof of learning separation according to our definitions in Section~\ref{sec:background}.
    \item[(ii)] We introduce a new approach to translating the learning separations of~\cite{liu:dlp, servedio:q_c_learnability} to comply with our definitions in Section~\ref{sec:background} by changing the underlying functions for which the hardness assumptions in~\cite{blum:hardness} and~\cite{kearns:clt} are sufficient.
    \end{itemize}
    
    \item We find that the learning separations in literature largely rely on the classical hardness of \textit{evaluating} the function generating the data, as opposed to the hardness of \textit{identifying} the function. 
    We discuss how the identification problem can be what is needed in practice, and we address this gap by proving two new learning separations where the classical hardness lies in identifying the function generating the data (see Theorems~\ref{thm:modexp} and~\ref{thm:elgamal}).
\end{itemize}
\end{itemize}

\noindent \textbf{\underline{Section 4}}:
\begin{itemize}
\item We show how leveraging stronger complexity-theoretic assumptions can lead to learning separations where the data is generated by a genuine quantum process. 
Our main contribution is Theorem~\ref{thm:seps_no-eff-data}, which outlines a generic method of establishing learning separations from $\mathsf{BQP}$-complete functions. 
We also provide two lemmas, Lemmas~\ref{lemma:1} and~\ref{lemma:2}, which introduce natural assumptions under which the criteria in Theorem~\ref{thm:seps_no-eff-data} are satisfied. 
Finally, we show how Theorem~\ref{thm:seps_no-eff-data} can be used to build learning separations from problems in quantum many-body physics.
\end{itemize}

\noindent \textbf{\underline{Section 5}}:
\begin{itemize}
\item To connect our work to some of the related results in the field~\cite{huang:science, huang:power, brien:ham_learning, haah:ham_learning}, we discuss selected topics related to learning separations with classical data:
\begin{itemize}
    \item[Section~\ref{subsec:huang}] We discuss the milestone work of Huang et al.~\cite{huang:science} and how their classical machine learning methods based on the classical shadow framework relate to learning separations with quantum-generated data (i.e., those from Theorem~\ref{thm:seps_no-eff-data}).
     In particular, we construct a family of Hamiltonians whose ground state properties cannot be predicted by any classical machine learning method based on cryptographic assumptions (see Theorem~\ref{thm:limitations_huang}).
     These results show that the conditions needed to achieve learnability in~\cite{huang:science} essentially cannot be relaxed.
    \item[Section~\ref{subsec:power_data}] We discuss a specific example (i.e., evaluating parameterized quantum circuits) that exemplifies how access to data radically enhances what is efficiently evaluated classically.
    \item[Section~\ref{subsec:proper_physics}] We discuss how two physically-motivated problems (i.e., Hamiltonian learning, and identifying order parameters and phases of matter) naturally fit in a learning setting where the learner is constrained to output a hypothesis from a fixed hypothesis class.
\end{itemize}
\end{itemize}

\setcounter{footnote}{0} 
\section{Formalizing quantum advantage in learning theory}
\label{sec:background}

In this section we introduce the required background and definitions.
Firstly, in Section~\ref{subsec:pac}, we discuss the relevant definitions from computational learning theory (for more details see~\cite{kearns:clt}), and we clarify the finer points regarding the possible definitions of a learning separation.
Afterwards, in Section~\ref{subsec:complexity}, we discuss the areas of complexity theory relevant to learning separations.

\subsection{Learning separations in the PAC learning framework}
\label{subsec:pac}

In this paper we use the standard terminology of the efficient \textit{probably approximately correct} (PAC) learning framework, and we focus on the supervised learning setting (for an overview of the generative modelling setting see~\cite{sweke:gen_mod}).
In this framework a learning problem is defined by a \textit{concept class} $\mathcal{C} = \{\mathcal{C}_n\}_{n \in \mathbb{N}}$, where each $\mathcal{C}_n$ is a set of \textit{concepts}, which are functions from some \textit{input space} $\mathcal{X}_n$ (in this paper we assume $\mathcal{X}_n$ is either $\{0,1\}^n$ or $\mathbb{R}^n$) to some \textit{label set} $\mathcal{Y}_n$ (in this paper we assume $\mathcal{Y}_n$ is $\{0,1\}$, with the exception of Section~\ref{subsec:new_sep} where it is $\{0,1\}^n$)\footnote{Since the focus of this paper is on the computational complexity of the learner, we choose to explicitly highlight the relevance of the instance size $n$ in our notation.}.
As input the learning algorithm has access to a procedure $EX(c, \mathcal{D}_n)$ (sometimes called an \textit{example oracle}) that runs in unit time, and on each call returns a labeled \textit{example} $(x, c(x))$, where $x \in \mathcal{X}_n$ is drawn according to \textit{target distributions} $\mathcal{D} = \{\mathcal{D}_n\}_{n \in \mathbb{N}}$.
Finally, the learning algorithm has associated to it a hypothesis class $\mathcal{H} = \{\mathcal{H}_n\}_{n \in \mathbb{N}}$, and its goal is to output a \textit{hypothesis} $h \in \mathcal{H}_n$ -- which is another function from $\mathcal{X}_n$ to $\mathcal{Y}_n$--  that is in some sense ``close'' to the concept $c \in \mathcal{C}_n$ generating the examples.

In the statistical version of the PAC learning framework the learning algorithm has to identify (and/or evaluate) a good hypothesis using $\mathcal{O}\left(\mathrm{poly}(n)\right)$ many queries to $EX(c, \mathcal{D}_n)$, and the computational complexity (i.e., ``runtime'') of the learning algorithm is not considered. 
In this paper however, we focus on the \emph{efficient} PAC learning framework, where the learning algorithm must output such a good hypothesis in \emph{time} $\mathcal{O}\left(\mathrm{poly}(n)\right)$ (note that this also implies that the learning algorithm can only use $\mathcal{O}\left(\mathrm{poly}(n)\right)$ many queries to $EX(c, \mathcal{D}_n)$).
Moreover, in this paper, we study exponential separations specifically with respect to the time complexity of the learning algorithms.

The PAC learning framework formalizes (binary-valued) supervised learning. 
For instance, in the learning scenario where one wants to detect a specific object in an image, the concepts are defined to attain the value $1$ when the object is present and 0 otherwise. 
Moreover, the oracle represents the set of training examples that is available in supervised learning.
We formally define efficient PAC learnability as follows.

\begin{definition}[Efficient probably approximately correct learnability]
\label{def:learnability}
A concept class $\mathcal{C} = \{\mathcal{C}_n\}_{n \in \mathbb{N}}$ is \textit{efficiently PAC learnable} under target distributions $\mathcal{D} = \{\mathcal{D}_n\}_{n \in \mathbb{N}}$ if there exists a hypothesis class $\mathcal{H} = \{\mathcal{H}_n\}_{n \in \mathbb{N}}$ and a (randomized) learning algorithm $\mathcal{A}$ with the following property: for every $c \in \mathcal{C}_n$, and for all $0 < \epsilon < 1/2$ and $0 < \delta < 1/2$, if $\mathcal{A}$ is given access to $EX(c, \mathcal{D}_n)$ and $(\epsilon, \delta)$, then with probability at least $1-\delta$ over the random examples drawn from $EX(c, \mathcal{D}_n)$ and over the internal randomization of $\mathcal{A}$, the learning algorithm $\mathcal{A}$ outputs a specification{\color{blue}\footnote{{The hypotheses (and concepts) are specified according to some enumeration $R: \cup_{n \in \mathbb{N}}\{0,1\}^n \rightarrow \cup_{n} \mathcal{H}_n$ (or, $\cup_{n} \mathcal{C}_n$) and by a ``specification of $h \in \mathcal{H}_n$'' we mean a string $\sigma \in \{0,1\}^*$ such that $R(\sigma) = h$ (see~\cite{kearns:clt} for more details).}}} of some $h \in \mathcal{H}_n$ that satisfies
\[
\mathsf{Pr}_{x \sim \mathcal{D}_n}\big[h(x) \neq c(x)\big] \leq \epsilon.
\]
Moreover, the learning algorithm $\mathcal{A}$ must run in time $\mathcal{O}(\mathrm{poly}(n, 1/\epsilon, 1/\delta))$.
If the learning algorithm is a polynomial-time classical algorithm (or, a quantum algorithm), we say that the concept class is \textit{classically learnable} (or, \textit{quantumly learnable}, respectively).

\end{definition}

An important thing to note in the above definition is that the learner itself consists of two parts:\ a hypothesis class, and a learning algorithm.
More precisely, the learner consists of a family of functions that it will use to approximate the concepts (i.e., the hypothesis class), and of a way to select which function from this family is the best approximation for a given concept (i.e., the learning algorithm).
This is generally not very different from how supervised learning is done in practice.
For example, in deep learning the hypothesis class consists of all functions realizable by a deep neural network with some given architecture, and the learning algorithm uses gradient descent to find the best hypothesis (i.e., the best assignment of weights).

As we will discuss in more detail in Section~\ref{sec:eff_data}, it is important to consider how the concept class is specified. 
What is important is that there must be some unambiguous definition of the concept class. 
In particular, the learnability criterion says that the concept class $\mathcal{C}$ is learnable if there exists an algorithm $\mathcal{A}_\mathcal{C}$ that satisfies the PAC criterion in Definition~\ref{def:learnability} with respect to $\mathcal{C}$. 
The point is that the algorithm  $\mathcal{A}_\mathcal{C}$ is tailored to the specific concept class  $\mathcal{C}$, and for this reason the concept class  $\mathcal{C}$ has to somehow be unambiguously specified and fixed. 
For example, it is not acceptable that we leave some ambiguity as to what the concept class is precisely (e.g., it is specified by some sequence of primes, but you are not told which one), and then conclude that the concept class is not learnable because the learning algorithm would have to work for all possibilities that the ambiguity leaves open (which would be equivalent to a concept class involving all sequences of primes).

With regards to the hypotheses that the learner is allowed to output, there are two settings to consider. 
Firstly, the learner can have unrestricted freedom and be able to output arbitrary hypotheses. 
Alternatively, the learner can be constrained to only output hypotheses from a fixed hypothesis class. 
In this paper, our main focus is on investigating the limitations of \textit{all} possible classical learners for a given task. 
To do so, we primarily focus on setting where the learner has the flexibility to output arbitrary hypotheses, barring certain tractability constraints which will be discussed in the paragraph below.
Our goal is to demonstrate separations that establish the inability of classical learners, regardless of the hypotheses they can output, to efficiently solve a learning problem that can be solved by a quantum learner. 
Nonetheless, in certain cases, it is natural to constrain the learner to only output hypotheses from a fixed hypothesis class. 
We explore this setting, along with an instance of it called proper PAC learning, in Sections~\ref{subsubsec:proper_pac} and~\ref{subsec:proper_physics}.

When the learner is allowed to output arbitrary hypotheses, it becomes necessary to limit the computational power of the hypotheses. 
Constraining the learning algorithm to run in polynomial-time turns out to be pointless if one allows arbitrary superpolynomial-time hypotheses. 
More precisely, if we allow superpolynomial-time hypotheses, then any concept class that can be learned by a superpolynomial-time learning algorithm, can also be learned by a polynomial-time learning algorithm (see Appendix~\ref{appendix:poly_eval} for more details). 
Intuitively, this is because by changing the hypotheses one can ``offload'' the learning algorithm onto the evaluation of the hypotheses, which makes any constraints on the learning algorithm pointless.
This is different when we constrain the learner to only be able to output hypotheses from a fixed hypothesis class, in which case it can be meaningful and natural to consider hypotheses with superpolynomial runtimes (see also Sections~\ref{subsubsec:proper_pac} and~\ref{subsec:proper_physics}).
In conclusion, if the learner is free to output arbitrary hypotheses, then we must make sure to restrict the learner to output efficiently evaluatable hypotheses~\cite{kearns:clt}.
Finally, because we are studying separations between classical and quantum learners, we make the distinction whether the hypotheses are efficiently evaluatable classically or quantumly.

\begin{definition}[Efficiently evaluatable hypothesis class]
    A hypothesis class $\mathcal{H} = \{\mathcal{H}_n\}_{n \in \mathbb{N}}$ is \emph{classically (quantumly) efficiently evaluatable} if there exists a classical (respectively quantum) polynomial-time evaluation algorithm $\mathcal{A}_{\mathrm{eval}}$ that on input  $x \in \mathcal{X}_n$ and a specification of a hypothesis $h \in \mathcal{H}_n$, outputs $\mathcal{A}_{\mathrm{eval}}(x, h) = h(x)$.
\end{definition}

For example, the hypotheses could be specified by a polynomial-sized Boolean circuit, in which case they are \textit{classically} efficiently evaluatable.
On the other hand, the hypotheses could also be specified by polynomial-depth quantum circuits, in which case they are \textit{quantumly} efficiently evaluatable.
If the family of quantum circuits that make up the hypothesis class is $\mathsf{BQP}$-complete, then the hypothesis class will be quantumly efficiently evaluatable, but not classically efficiently evaluatable (assuming $\mathsf{BPP}\neq\mathsf{BQP}$).
In this paper we will drop the ``efficiently'' and simply call a  hypothesis class classically- or quantumly evaluatable.

Given the definitions above one may assume that there is only one way to define a learning separation in the PAC learning framework.
However, it is in fact more subtle, and there are various definitions that each have operationally different meanings.
In particular, one needs to differentiate whether the learning algorithm is classical or quantum, and whether the hypothesis class is classically- or quantumly- evaluatable. 
As a result, we can consider four categories of learning problems: concept classes that are either \textit{classically-} or \textit{quantumly-} \textit{learnable} (i.e., whether the learning algorithm is classical or quantum), using a \textit{classically-} or \textit{quantumly-} evaluatable hypothesis class.
We denote these categories by $\mathsf{CC}, \mathsf{CQ}, \mathsf{QC}$, and $\mathsf{QQ}$, where the first letter signifies whether the concept class is classically- or quantumly- learnable , and the second letter signifies whether the learner uses a classically- or quantumly- evaluatable hypothesis class.
These distinctions are \emph{not} about the nature of the data (i.e., we only consider the setting where the examples are classical) as it often occurs in literature, and even on the \href{https://en.wikipedia.org/wiki/Quantum_machine_learning}{Wikipedia-page} of quantum machine learning.

\begin{definition}[Categories of learning problem]\hspace{0pt}
\label{def:categories}
\begin{itemize}
    \item Let $\mathsf{CC}$ denote the set of tuples $\big(\mathcal{C}, \mathcal{D}\big)$ such that~$\mathcal{C}$ is $\mathsf{\textbf{classically}}$ \textit{learnable} under target distributions $\mathcal{D}$ with a $\mathsf{\textbf{classically}}$ \textit{evaluatable} hypothesis class.
    
    \item Let $\mathsf{CQ}$ denote the set of tuples $\big(\mathcal{C}, \mathcal{D}\big)$ such that $\mathcal{C}$ is $\mathsf{\textbf{classically}}$ \textit{learnable} under target distributions $\mathcal{D}$ with a $\mathsf{\textbf{quantumly}}$ \textit{evaluatable} hypothesis class.
    
    \item Let $\mathsf{QC}$ denote the set of tuples $\big(\mathcal{C}, \mathcal{D}\big)$ such that $\mathcal{C}$ is $\mathsf{\textbf{quantumly}}$ \textit{learnable} under target distributions $\mathcal{D}$ with a $\mathsf{\textbf{classically}}$ \textit{evaluatable} hypothesis class.
    
    \item Let $\mathsf{QQ}$ denote the set of tuples $\big(\mathcal{C}, \mathcal{D}\big)$ such that $\mathcal{C}$ is $\mathsf{\textbf{quantumly}}$ \textit{efficiently learnable} under target distributions $\mathcal{D}$ with a $\mathsf{\textbf{quantumly}}$ \textit{evaluatable} hypothesis class.
    
\end{itemize}
\end{definition}

We remark that our definitions do not (yet) talk about the computational tractability of the concepts, the importance of which we will discuss in Section~\ref{subsec:complexity} and throughout Sections~\ref{sec:eff_data} and~\ref{sec:no_eff_data}.
We now proceed with a few observations.
Firstly, since any classical algorithm can be simulated by a quantum algorithm it is clear that $\mathsf{CC} \subseteq \mathsf{CQ}$, $\mathsf{CC} \subseteq \mathsf{QC}$, $\mathsf{CC} \subseteq \mathsf{QQ}$, $\mathsf{CQ} \subseteq \mathsf{QQ}$, and $\mathsf{QC} \subseteq \mathsf{QQ}$.
Secondly, if the hypothesis class is quantumly evaluatable, then it does not matter whether we constrain the learning algorithm to be a classical- or a quantum- algorithm.
More precisely, any learning problem that is quantumly learnable using a quantumly evaluatable hypothesis class is also classically learnable using \textit{another} quantumly evaluatable hypothesis class.
This observation is summarized in the lemma below, the proof of which is deferred to Appendix~\ref{appendix:cqqq}.

\begin{restatable}{lemma}{CQequalsQQ}
\label{lemma:cq=qq}
$\mathsf{CQ} = \mathsf{QQ}$.
\end{restatable}

The above lemma is analogous to why we constrain the hypotheses to be efficiently evaluatable, in the sense that by changing the hypothesis class one can ``offload'' the \textit{quantum} learning algorithm onto the evaluation of the \textit{quantum} hypotheses.
We reiterate that it is critical that one can change the hypothesis class when mapping a learning problem in $\mathsf{QQ}$ to $\mathsf{CQ}$.
If the learner is constrained to output hypotheses from a fixed hypothesis class, then such a collapse does not happen.

Having studied the relations between the categories learning problems, we can now specify what it means for a learning problem to exhibit a separation between classical and quantum learners.

\begin{definition}[Learning separation]
\label{def:separations}
A learning problem $L = \big(\mathcal{C}, \mathcal{D}\big)$ is said to exhibit a 

\begin{itemize}
    \item $\mathsf{CC}/\mathsf{QC}$ separation if $L \in \mathsf{QC}$ and $L \not\in \mathsf{CC}$.
    
    \item $\mathsf{CC}/\mathsf{QQ}$ separation if $L \in \mathsf{QQ}$ and $L \not\in \mathsf{CC}$.
\end{itemize}

\end{definition}

Firstly, note that due to the previously listed inclusions any $\mathsf{CC/QC}$ separation is also a $\mathsf{CC/QQ}$ separation.
Secondly, note that by fully relying on the classical intractability of concepts one can construct trivial learning separations that are less about ``learning'' in an intuitive sense.
More precisely, consider the separation exhibited by the concept class $\mathcal{C} = \{\mathcal{C}_n\}_{n \in \mathbb{N}}$, where each $\mathcal{C}_n$ consists of a \textit{single} concept that is classically hard to evaluate on a fraction of inputs even in the presence of data, yet it can be efficiently evaluated by a quantum algorithm.
This singleton concept class is clearly quantumly learnable using a quantumly evaluatable hypothesis class.
Also, it is not classically learnable using any classically evaluatable hypothesis class, since this would violate the classical intractability of the concepts.
However, note that the quantum learner \textit{requires no data} to learn the concept class, so it is hard to argue that this is a genuine learning problem. 
We will discuss how to construct examples of such concept classes in Sections~\ref{subsec:dlp} and~\ref{subsec:3root}.

\begin{observation}[Trivial learning separation without data] 
    \label{lemma:trivial_sep}
    Consider a family of concept classes $\mathcal{C} = \{\mathcal{C}_n\}_{n \in \mathbb{N}}$, where each $\mathcal{C}_n = \{c_n\}$ consists of a single concept that is classically hard to evaluate on a fraction of inputs when given access to examples, yet it can be efficiently evaluated by a quantum algorithm.
    Then, $\mathcal{C}$ exhibits a $\mathsf{CC/QQ}$ separation which is quantum learnable without requiring data.
\end{observation}

We want to emphasize that some concept classes are \textit{efficiently evaluatable on a classical computer}, yet they are \textit{not classically learnable}. 
One such example is the class of polynomially-sized logarithmic-depth Boolean circuits~\cite{kearns:clt}. 
Moreover, in Section~\ref{subsec:new_sep}, we provide an example of concept class which (assuming a plausible but relatively unexplored hardness assumption) exhibits a $\mathsf{CC/QC}$ separation where the concepts are efficiently evaluatable on a classical computer.

\subsubsection{Learning separations with a fixed hypothesis class and proper PAC learning}
\label{subsubsec:proper_pac}

In some practical settings, it can be natural to constrain the learner to only output hypotheses from a fixed hypothesis class.
To give a physics-motivated example, when studying phases of matter one might want to identify what observable properties characterize a phase.
One can formulate this problem as finding a specification of the correct hypothesis selected from a hypothesis class consisting of possible \emph{order parameters}.
More precisely, we fix the hypotheses to be of a particular form, e.g., those that compute certain expectation values of ground states given a specification of a Hamiltonian\footnote{Note the computation of these hypotheses can be $\mathsf{QMA}$-hard, as it involves preparing ground states.  Nonetheless, we can still study whether a learner is able to \textit{identify} which of these hypotheses matches the data.}.
We further discuss this setting of characterizing phases of matter in Section~\ref{subsec:proper_physics}, where we also discuss Hamiltonian learning as a natural setting in which the learner is constrained to output hypotheses from a fixed hypothesis class.

Recall that in the standard PAC learning framework discussed in the previous section, the learner is free to output arbitrary hypotheses (barring tractability constraints discussed in Appendix~\ref{appendix:poly_eval}).
It therefore fails to capture the setting where one aims to characterize phases of matter, as the learner might output hypotheses that are not order parameters, which will not allow one to identify physical properties that characterize a phase.
To remedy this, one could consider the setting where the learner is constrained to output hypotheses from a fixed hypothesis class.

\begin{definition}[Efficient PAC learnability with fixed hypothesis class]
\label{def:proper_learnability}
A concept class $\mathcal{C} = \{\mathcal{C}_n\}_{n \in \mathbb{N}}$ is \textit{efficiently PAC learnable with a fixed hypothesis class} $\mathcal{H} = \{\mathcal{H}_n\}_{n \in \mathbb{N}}$ under target distributions $\mathcal{D} = \{\mathcal{D}_n\}_{n \in \mathbb{N}}$ if there exists a (randomized) learning algorithms $\mathcal{A}$ with the following property: for every $c \in \mathcal{C}_n$, and for all $0 < \epsilon < 1/2$ and $0 < \delta < 1/2$, if $\mathcal{A}$ is given access to $EX(c, \mathcal{D}_n)$ and $\epsilon$ and $\delta$, then with probability at least $1 - \delta$, $\mathcal{A}$ outputs a specification of some $h \in \mathcal{H}_n$ that satisfies
\[
\mathsf{Pr}_{x \sim \mathcal{D}_n}\big[h(x) \neq c(x)\big] \leq \epsilon.
\]
Moreover, the learning algorithm $\mathcal{A}$ must run in time $\mathcal{O}(\mathrm{poly}(n, 1/\epsilon, 1/\delta))$.
\end{definition}

In the above definition, the probability $1-\delta$ is over the random examples from $EX(c, \mathcal{D}_n)$ and over the internal randomization of $\mathcal{A}_n$.
If the learning algorithm is a polynomial-time classical algorithm (or, a quantum algorithm), we say that the concept class is \textit{classically learnable with fixed hypothesis class} (or, \textit{quantumly learnable with fixed hypothesis class}, respectively).
An example of learning with a fixed hypothesis class is that of \emph{proper PAC learning}. 
In proper PAC learning the learner is constrained to only output hypothesis from the concept class it is trying to learn.

We emphasize again that if the learner is constrained to output hypotheses from a fixed hypothesis class, then it is allowed and reasonable for the hypothesis class to be (classically- or quantumly-) intractable.
In particular, doing so will not trivialize the definitions as it did in the standard PAC learning framework (see Appendix~\ref{appendix:details}) as this requires one to be able to change the hypotheses.

In the setting where the learner is constrained to output hypotheses from a fixed hypothesis class, it is relatively clear how to define a learning separation.
In particular, one only has to distinguish whether the learning algorithm is an efficient classical- or quantum- algorithm, which we capture by defining the following categories of learning problems.

\begin{definition}[Categories of learning problem -- fixed hypothesis class $\mathcal{H}$]\hspace{0pt}\\
\vspace{-15pt}
\begin{itemize}
    \item Let $\mathsf{C}_{\mathcal{H}}$ denote the set of tuples $\big(\mathcal{C}, \mathcal{D}\big)$ such that $\mathcal{C}$ is $\mathsf{\textbf{classically}}$ \textit{learnable with fixed hypothesis class $\mathcal{H}$} under target distributions $\mathcal{D}$.
    
    \item Let $\mathsf{Q}_{\mathcal{H}}$ denote the set of tuples $\big(\mathcal{C}, \mathcal{D}\big)$ such that $\mathcal{C}$ is $\mathsf{\textbf{quantumly}}$ \textit{learnable with fixed hypothesis class $\mathcal{H}$} under target distributions $\mathcal{D}$.
\end{itemize}
\end{definition}

We can now specify what it means for a learning problem to exhibit a separation between classical and quantum learners in the setting where the learner is constrained to output hypotheses from a fixed hypothesis class.

\begin{definition}[Learning separation -- fixed hypothesis class $\mathcal{H}$]\hspace{0pt}\\
\label{def:proper_separations}

\vspace{-10pt} 
\noindent A learning problem $L = \big(\mathcal{C}, \mathcal{D}\big) \in \mathsf{Q}_{\mathcal{H}}$ is said to exhibit a $\mathsf{C}_{\mathcal{H}}/\mathsf{Q}_{\mathcal{H}}$ separation if $L \not\in \mathsf{C}_{\mathcal{H}}$.
\end{definition}

In Sections~\ref{subsec:new_sep} and~\ref{subsec:elgamal}, we provide examples of learning separations in the setting where the learner is constrained to output hypotheses from a fixed hypothesis class.
Moreover, in Section~\ref{subsec:proper_physics}, we further discuss the practical relevance of this setting by discussing how it captures certain physics-motivated examples of learning settings.

\subsubsection{Identification versus Evaluation} 
\label{subsubsec:identify_evaluate}

An important difference in what exactly entails a learning task in practice is whether the learner has to only \emph{identify} a hypothesis that is close to the concept generating the examples, or whether the learner also has to \emph{evaluate} the hypothesis on unseen examples later on.
Moreover, these differences in tasks have implications for the role of quantum computers in achieving separations.
This difference in tasks is reflected in two aspects within the definitions discussed in this section.

Firstly, this difference in tasks is reflected in the difference between $\mathsf{CC}/\mathsf{QQ}$ and $\mathsf{CC}/\mathsf{QC}$ separations.
In particular, it is reflected in the task that requires a quantum computer (i.e., what task needs to be classically intractable yet efficient on a quantum computer).
On the one hand, for a $\mathsf{CC}/\mathsf{QC}$ separation, one has to show that only a quantum algorithm can \textit{identify} how to label unseen examples using a classical algorithm.
On the other hand, for a $\mathsf{CC}/\mathsf{QQ}$ separation, one also needs to show that only a quantum algorithm can \textit{evaluate} the labels of unseen examples.
In Section~\ref{subsec:new_sep}, we provide an example of a $\mathsf{CC/QC}$ separation (contingent on a plausible though relatively unexplored hardness assumption), where the classical hardness lies in \textit{identifying} an hypothesis matching the examples, since the concepts are efficiently evaluatable classically.

Secondly, the difference in tasks is also reflected in the difference between the setting where the learner is allowed to output arbitrary hypothesis, or whether it can only output hypotheses from a fixed hypothesis class.
In the arbitrary hypothesis class setting, one has to demand that the hypotheses are efficiently evaluatable (i.e., see Appendix~\ref{appendix:details}), which allows the learner to efficiently \textit{evaluate} the hypotheses on unseen examples.
In the fixed hypothesis class setting, the hypotheses need not be efficiently evaluatable, and the learner is only required to \textit{identify} the correct hypothesis without having to evaluate it on unseen examples.
In Sections~\ref{subsec:new_sep} and~\ref{subsec:elgamal}, we provide examples of separation in the setting where the learner is constrained to output hypotheses from a fixed hypothesis class.
Note that the classical hardness in these separation lies identifying the hypotheses, as we do not require the learner to evaluate the hypothesis on unseen examples afterwards.

\subsection{Complexity theory}
\label{subsec:complexity}

In this section we provide a short overview of the areas of complexity theory that we will refer to when discussing separations in the PAC learning framework.
In particular, we focus on the computational hardness assumptions that one can leverage to establish a learning separation.

We turn our attention to the definition of the PAC learning framework (see Definition~\ref{def:learnability}) and make some observations that will be relevant later.
First, we note that the hypothesis that the learning algorithm outputs is only required to be correct with probability $\epsilon$ over the target distribution.
In complexity theory, this is related to the notion of heuristic complexity classes (for more details see~\cite{bogdanov:average}).
To define heuristic complexity classes, we first need to incorporate the target distribution as a part of the problem, which is done by considering distributional problems.

\begin{definition}[Distributional problem]
\label{def:distr_problem}
A distributional problem $(L, \mathcal{D})$ consists of a language $L\subseteq\{0,1\}^*$\footnote{Throughout this paper, we also use an equivalent definition of a language $L\subseteq \{0,1\}^*$ by instead calling it a \emph{problem} and defining it as a function $L: \{0, 1\}^* \rightarrow \{0,1\}$ such that $L(x) = 1$ if and only if $x \in L$.} and a family of distributions $\mathcal{D} = \{\mathcal{D}_n\}_{n \in \mathbb{N}}$ such that $\mathrm{supp}(\mathcal{D}_n) \subseteq \{0,1\}^n$.
\end{definition}

\noindent Having defined distributional problems, we now define the relevant heuristic complexity classes.

\begin{definition}[Heuristic complexity~\cite{bogdanov:average}]
\label{def:heurbpp}
A distributional problem $(L, \mathcal{D})$ is in $\mathsf{HeurBPP}$ if there exists a polynomial-time randomized classical algorithm $\mathcal{A}$\footnote{More precisely, a Turing machine.} such that for all $n$ and $\epsilon > 0$\footnote{Here $0^{\lfloor 1/\epsilon \rfloor}$ denotes the bitstring consisting of $\lfloor1/\epsilon\rfloor$ zeroes (i.e., it is a unary specification of the precision $\epsilon$).}:
\begin{align}
\label{eq:heur}
     \mathsf{Pr}_{x \sim \mathcal{D}_n}\Big[\mathsf{Pr}\big(\mathcal{A}(x, 0^{\lfloor 1/\epsilon \rfloor}) = L(x)\big) \geq \frac{2}{3} \Big] \geq 1 - \epsilon,
\end{align}
where the inner probability is taken over the internal randomization of $\mathcal{A}$.\\
\indent Analogously, we say that a distributional problem $(L, \{\mathcal{D}_n\}_{n \in \mathbb{N}})$ is in $\mathsf{HeurBQP}$ if there exists a polynomial-time \textit{quantum} algorithm $\mathcal{A}$ that satisfies the property in Eq.~\eqref{eq:heur}.
\end{definition}

A related and perhaps better known area of complexity theory is that of \textit{average-case} complexity.
The main difference between average-case complexity and heuristic complexity, is that in the latter one is allowed to err, whereas in the former one can never err but is allowed to output ``don't know''.
Note that an average-case algorithm can always be converted into a heuristic algorithm by simply outputting a random result instead of outputting ``don't know''.
Similarly, if there is a way to efficiently check if a solution is correct, any heuristic algorithm can be turned into an average-case algorithm by outputting ``don't know'' when the solution is incorrect.
Even though they are closely related, in the PAC learning framework one deals with heuristic complexity.

While heuristic-hardness statements are not as common in quantum computing literature, many cryptographic security assumptions (such as that of RSA and Diffie-Hellman) are in fact examples of heuristic-hardness statements.
These heuristic-hardness statements are generally derived from \textit{worst-case to average-case reductions}, which show that being correct with a certain probability over a specific input distribution is at least as difficult as being correct on all inputs. 
Problems that admit a worst- to average-case reduction are called \emph{random self-reducible} (for a formal definition see~\cite{feigenbaum:rsr}).
It is worth noting that despite the term ``average-case'', these reductions can also yield heuristic hardness statements.
Specifically, if one can efficiently check whether a solution is correct, then a worst-case to average-case reduction also results in a heuristic hardness statement when the worst-case is hard.
For instance, a worst-case to average-case reduction by Blum and Micali~\cite{blum:hardness} demonstrates that for the discrete logarithm problem being correct on any $\frac{1}{2} + \frac{1}{\mathrm{poly}(n)}$ fraction of inputs is as difficult as being correct for all inputs (notably, modular exponentiation allows for efficient checking of the correctness of a discrete logarithm solution).

Finally, there is the notion of the example oracle.
The fact that access to the example oracle radically enhances what can be efficiently evaluated is related (though not completely analogous, as we will explain below) to the notion of ``advice'' complexity classes such as $\mathsf{P/poly}$.
\begin{definition}[Polynomial advice~\cite{arora:book}]
\label{def:p/poly}
A problem $L: \{0,1\}^* \rightarrow \{0,1\}$ is in $\mathsf{P/poly}$ if there exists a polynomial-time classical algorithm $\mathcal{A}$ with the following property: for every $n$ there exists an advice bitstring $\alpha_n \in \{0,1\}^{\mathrm{poly}(n)}$ such that for all $x \in \{0,1\}^n$:
\begin{align}
\label{eq:advice}
     \mathcal{A}(x, \alpha_n) = L(x).
\end{align} 

Analogously, we say that a problem $L$ is in $\mathsf{BQP/poly}$ if there exists a polynomial-time quantum algorithm $\mathcal{A}$ with the following property: for every $n$ there exists an advice bitstring $\alpha_n~\in~\{0,1\}^{\mathrm{poly}(n)}$ such that for all $x \in \{0,1\}^n$:
\begin{align}
\mathsf{Pr}\big(\mathcal{A}(x, \alpha_n) = L(x)\big) \geq \frac{2}{3},
\end{align}
where the probability is taken over the internal randomization of $\mathcal{A}$.
\end{definition}

Equivalently, one could also define $\mathsf{P/poly}$ as the class of problems solvable by a \textit{non-uniform} family of polynomial-size Boolean circuits (i.e., there could be a completely different circuit for each input length).
Also, since in the PAC learning framework we deal with randomized learning algorithms one may want to consider $\mathsf{BPP/poly}$ instead, however by~\cite{adleman:theorem} we have that $\mathsf{BPP}\subseteq \mathsf{P/poly}$, and so $\mathsf{BPP/poly} = \mathsf{P/poly}$.
On the other hand, from the perspective of the PAC learning framework, it is both natural and essential to allow the algorithm that uses the advice to err on a fraction of inputs, which is captured by the complexity class $\mathsf{HeurP/poly}$.  

\begin{definition}[Heuristic complexity with polynomial advice]
\label{def:heurp/poly}
A distributional problem $(L, \mathcal{D})$ is in $\mathsf{HeurP/poly}$ if there exists a polynomial-time classical algorithm $\mathcal{A}$ with the following property: for every $n$ and $\epsilon > 0 $ there exists an advice string $\alpha_{n, \epsilon} \in \{0,1\}^{\mathrm{poly}(n, 1/\epsilon)}$ such that:
\begin{align}
\label{eq:heur_advice}
     \mathsf{Pr}_{x \sim \mathcal{D}_n}\left[\mathcal{A}(x, 0^{\lfloor 1/\epsilon \rfloor}, \alpha_{n, \epsilon}) = L(x)\right] \geq 1 - \epsilon.
\end{align} 

Analogously, we say that $(L, \mathcal{D})$ is in $\mathsf{HeurBQP/poly}$ if there exists a polynomial-time quantum algorithm $\mathcal{A}$ such that: for every $n$ and $\epsilon > 0 $ there exists an advice string $\alpha_{n, \epsilon} \in \{0,1\}^{\mathrm{poly}(n, 1/\epsilon)}$ with the following property:
\begin{align}
\label{eq:heur_advice1}
     \mathsf{Pr}_{x \sim \mathcal{D}_n}\left[ \mathsf{Pr}\left(\mathcal{A}(x, 0^{\lfloor 1/\epsilon \rfloor}, \alpha_{n, \epsilon}) = L(x) \right) \geq \frac{2}{3} \right] \geq 1 - \epsilon,
\end{align} 
where the inner probability is taken over the internal randomization of $\mathcal{A}$.
\end{definition}

Note that in the PAC learning framework, the advice that the learning algorithm gets is of a specific form, namely that obtained through queries to the example oracle.
This is more closely related to the notion of ``sampling advice'' complexity classes such as $\mathsf{BPP/samp}$~\cite{huang:power} defined below.
In~\cite{huang:power} it is shown that $\mathsf{BPP/samp} \subseteq \mathsf{P/poly}$, i.e., sampling advice is not more powerful than the standard notion of advice.

\begin{definition}[Sampling advice~\cite{huang:power}]
\label{def:bpp/samp}
A problem $L: \{0,1\}^* \rightarrow \{0,1\}$ is in $\mathsf{BPP/samp}$ if there exists polynomial-time classical randomized algorithms $\mathcal{S}$ and $\mathcal{A}$ such that for every $n$:
\begin{itemize}
    \item $\mathcal{S}$ generates random instances $x \in \{0,1\}^n$  sampled from the distribution~$\mathcal{D}_n$.
    \item $\mathcal{A}$ receives as input $\mathcal{T} = \{(x_i, L(x_i)) \mid x_i \sim \mathcal{D}_n\}_{i=1}^{\mathrm{poly}(n)}$ and satisfies for all $x \in \{0,1\}^n$:
\begin{align}
\label{eq:samp}
    \mathsf{Pr}\big(\mathcal{A}(x, \mathcal{T}) = L(x)\big) \geq \frac{2}{3},
 \end{align} 
 where the probability is taken over the internal randomization of $\mathcal{A}$ and $\mathcal{T}$. 
\end{itemize}
\end{definition}

Having related notions in the PAC learning framework to different areas of complexity theory, we are now ready to determine what computational hardness assumptions one can leverage to establish that no classical learner is able to learn a given concept class.
More specifically, how hard must evaluating the concepts be for the concept class to not be classically learnable?
Since the learning algorithm is a \emph{randomized} algorithm that \emph{heuristically} computes the concepts when provided with \emph{advice} in the form of samples from the example oracle, the existence of a polynomial-time learning algorithm puts the concepts in a complexity class that we call $\mathsf{HeurBPP/samp}$.

\begin{definition}
\label{def:heurbpp/samp}
A distributional problem $(L, \mathcal{D})$ is in $\mathsf{HeurBPP/samp}$ if there exists classical randomized algorithms $\mathcal{S}$ and $\mathcal{A}$ such that for every $n$:
\begin{itemize}
    \item $\mathcal{S}$ generates random instances $x \in \{0,1\}^n$  sampled from the distribution~$\mathcal{D}_n$.
    \item $\mathcal{A}$ receives as input $\mathcal{T} = \{(x_i, L(x_i)) \mid x_i \sim \mathcal{D}_n\}_{i=1}^{\mathrm{poly}(n)}$ and for every $\epsilon>0$ satisfies:
\begin{align}
\label{eq:heursamp}
    \mathsf{Pr}_{x \sim \mathcal{D}_n}\left[\mathsf{Pr}\big(\mathcal{A}(x, 0^{\lfloor 1/\epsilon \rfloor}, \mathcal{T}) = L(x)\big) \geq \frac{2}{3}\right] \geq 1 - \epsilon,
 \end{align} 
 where the inner probability is taken over the internal randomization of $\mathcal{A}$ and $\mathcal{T}$.
\end{itemize}
\end{definition}

More precisely, if the concepts lie outside of $\mathsf{HeurBPP/samp}$, then the concept class is not classically learnable.
We can connect the class $\mathsf{HeurBPP/samp}$ to other complexity classes by adopting a proof strategy similar to that of~\cite{huang:power} (we defer the proof to Appendix~\ref{appendix:samp_poly}).

\begin{restatable}{lemma}{samppoly}
\label{lemma:samppoly}
$\mathsf{HeurBPP/samp} \subseteq \mathsf{HeurP/poly}$.
\end{restatable}

By the above lemma, we find that any problem not in $\mathsf{HeurP/poly}$ is also not in $\mathsf{HeurBPP/samp}$.
Consequently, to show the non-learnability of a concept class, it is sufficient to show that the concept class includes concepts that are not in $\mathsf{HeurP/poly}$.

\bigskip

Having discussed the related notions from computational learning theory and complexity theory, we are set to investigate how one establishes learning separations.
First, in Section~\ref{sec:eff_data}, we will analyze how existing learning separations have used efficient data generation, and we generalize this construction to $(i)$ establish a learning separation (contingent on a plausible though relatively unexplored hardness assumption) with efficiently evaluatable concepts, and $(ii)$ establish a learning separation in the setting where the learner is constrained to output an hypothesis from a fixed hypothesis class.
Afterwards, in Section~\ref{sec:no_eff_data}, we discuss the additional constructions required to prove separations in tune with the folklore that quantum machine learning is most likely to have it advantages when the data generated by a ``genuine quantum process''.
For an overview of the learning separations discussed throughout this paper see Table~\ref{table:separations}.

\begin{table}[h!]
\centering
\begin{tabular}{ c|c|c|c }
\textbf{First proposed in} & \textbf{Concepts based on} &\textbf{Separation} & \textbf{Complexity of concepts}\footnote{Assuming common assumptions in complexity theory and cryptography.} \\
\hline
\cite{liu:dlp} & Discrete logarithm & $\mathsf{CC/QQ}$ & $\not \in \mathsf{BPP}$\\
\hline
\cite{servedio:q_c_learnability} & Discrete cube root & $\mathsf{CC/QC}$ & $\not \in \mathsf{BPP}$ but $\in \mathsf{P/poly}$\\
\hline
Section~\ref{subsec:new_sep} & Modular exponentiation & $\mathsf{CC/QC}$ & $\in \mathsf{P}$\\
\hline
Section~\ref{subsec:elgamal} & Discrete cube root & $\mathsf{C}_{\mathcal{H}}/\mathsf{Q}_{\mathcal{H}}$ & $ \in \mathsf{P}$\\
\hline
Section~\ref{subsec:physical_systems} & Genuine quantum process & $\mathsf{CC/QQ}$ & $ \not\in \mathsf{HeurP/poly}$ but $\in \mathsf{BQP}$\\
\hline
\end{tabular}
\caption{An overview of the characteristics of the learning separations discussed in Section~\ref{sec:eff_data} and Section~\ref{sec:no_eff_data}.
Importantly, in Section~\ref{sec:no_eff_data} we establish learning separations extending beyond the cryptographic problems discussed in Section~\ref{sec:eff_data} to encompass essentially all $\mathsf{BQP}$-complete problems.}
\label{table:separations}
\end{table}

\setcounter{footnote}{0} 
\section{Learning separations with efficient data generation}
\label{sec:eff_data}

A commonality between the learning separations of~\cite{liu:dlp, servedio:q_c_learnability} is that the proof of classical non-learnability relies on the fact that the examples can be efficiently generated classically (i.e., the example oracle can be efficiently simulated classically)\footnote{The notion of efficiently generatable examples is closely related to the notion of \textit{random verifiability}~\cite{arrighi:blind}.}.
This is crucial, since it ensures that access to the example oracle does not enhance what a classical learner can evaluate relative to a conventional (non-learning) classical algorithm.
This then allows one to directly deduce classical non-learnability from a complexity-theoretic hardness assumption related to the concepts, since the existence of an efficient classical learner would imply the existence of an efficient classical  algorithm.
A similar observation was made by the authors of~\cite{perez:relation} (which came out in between~\cite{gyurik:old} and this paper), where they also study the problem of distribution-independent learning separations.

In this section we study the learning separations of~\cite{liu:dlp, servedio:q_c_learnability}, and we characterize them with respect to the type of learning separation they achieve (as discussed in Section~\ref{subsec:pac}), and the kind of hardness assumptions they leverage to obtain classical non-learnability (as discussed in Section~\ref{subsec:complexity}).
Firstly, in Section~\ref{subsec:dlp}, we discuss the $\mathsf{CC/QQ}$ separation of the discrete logarithm concept class of~\cite{liu:dlp}, whose concepts are believed to be classically intractable, no matter how they are specified.
Secondly, in Section~\ref{subsec:3root}, we discuss the $\mathsf{CC/QC}$ separation of the cube root concept class of~\cite{kearns:clt, servedio:q_c_learnability}, whose concepts are specified in a way that makes them classically intractable, though when specified in a different way they become classically efficient (i.e., the concepts are ``obfuscated'' versions of classically efficient functions).

While discussing the learning separations of~\cite{liu:dlp, servedio:q_c_learnability}, we will provide a fine grained analysis of the computational hardness assumptions necessary for establishing these separations.
It is crucial to be precise about the specific computational hardness assumptions, as subtle details significantly impact the types of learning separations achievable.
For instance, while the concept class based on the discrete logarithm from~\cite{liu:dlp} suggests a separation based on the canonical hardness assumption in~\cite{blum:hardness}, our fine-grained analysis reveals that this is not directly supported by the provided proof in~\cite{liu:dlp}. 
We address this by introducing a stronger modified hardness assumption and we show that this modification then suffices for a formal proof of classical non-learnability.
Nevertheless, we will also show that similar learning separations are attainable assuming the standard hardness assumption in~\cite{blum:hardness}.
However, these assumptions require the introduction of a new set of concepts based on a novel construction, where information about the base of the discrete logarithm is ``leaked'' through examples in the dataset.
Our detailed examination of the learning separation in~\cite{servedio:q_c_learnability} reveals a similar scenario, where we again could not recover a complete proof of classical non-learnability by leveraging the standard hardness assumptions in~\cite{kearns:clt}, based on the details in~\cite{servedio:q_c_learnability}.
We again address this issue in two ways: first, by introducing stronger hardness assumptions, which allow us to construct a proof of separation following the original approach of~\cite{servedio:q_c_learnability}, and second, by introducing a new construction of the concept class, for which it is possible to prove hardness based on the standard hardness assumption in~\cite{kearns:clt}.

In our fine-grained analysis, we also explore whether separations can be achieved for singleton concept classes, a crucial aspect in distinguishing genuine learning problems from problems that are just computational problems in disguise (as discussed in Section~\ref{sec:background}).
We delve into how different hardness assumptions can lead to learning separations for singleton concept classes with different characteristics, revealing a trade-off.
For example, while the computational hardness assumption we introduced to recover the learning separation of~\cite{liu:dlp} yields a $\mathsf{CC/QQ}$ separation for a binary singleton concept class, for the case of the typical hardness assumption in~\cite{blum:hardness} we could only achieve a $\mathsf{CC/QQ}$ separation for multi-valued (i.e., \textit{non-binary}) singleton concept class.
Similarly, we note a trade-off in learning separations based on the discrete cube root assumption, as discussed in Sections~\ref{subsec:3root}-\ref{subsec:elgamal}.
The computational hardness assumption introduced to recover the learning separation in~\cite{servedio:q_c_learnability} leads to a $\mathsf{CC/QC}$ separation for a singleton concept class, whereas our novel construction demonstrates that the typical hardness assumption in~\cite{blum:hardness} results in a $\mathsf{CC/QQ}$ separation for a singleton concept class (i.e., a tradeoff in the type of separation).
For an overview of these trade-offs, we refer to Table~\ref{table:assumptions_finedlp} \& \ref{table:assumptions_finedcr}.

Finally, while discussing the learning separations of~\cite{liu:dlp, servedio:q_c_learnability}, we notice that their proofs largely rely on the classical difficulty of \textit{evaluating} the hypotheses on unseen examples, rather than the difficulty of \textit{identifying} a hypothesis that is close to the concept generating the examples. 
To complement these works, we present two new examples of learning separations where the classical hardness lies in \textit{identifying} the concept that is generating the examples.
Specifically, in Section~\ref{subsec:new_sep}, we provide an example of a $\mathsf{CC/QC}$ separation (contingent on a plausible though relatively unexplored hardness assumption) where the concepts are classically efficiently evaluatable, making it impossible for the classical hardness to come from evaluating them on unseen examples.
Afterwards, in Section~\ref{subsec:elgamal}, we provide an example of a separation in the setting where the learner is constrained to output hypotheses from a fixed hypothesis class, in which case the learner is only required to identify the concept generating the examples, therefore also eliminating the possibility that the classical hardness comes from evaluating them on unseen examples\footnote{We remark that the concept class of Section~\ref{subsec:new_sep} also exhibits a separation in the setting the learner is constrained to output a hypothesis from a fixed hypothesis class. However, we choose to present it as a $\mathsf{CC/QC}$ separation to highlight that such separations are still possible if the concepts are classically efficiently evaluatable. Moreover, we still include the separation in the setting the learner is constrained to output a hypothesis from a fixed hypothesis class of Section~\ref{subsec:elgamal}, because it is not contingent on a relatively unexplored hardness assumption.}.

\begin{table}[h!]
\centering
\begin{tabular}{ c|c|c|c|c}
\diagbox{\textbf{Concepts}}{\textbf{Properties}} & \textbf{Hardness assumption} & \textbf{Separation} & \underline{\textbf{Binary?}} & \textbf{Singleton?} \\[0.3em]
\hline
\makecell{Fix sequence $\{(p_n, a_n)\}_{n \in \mathbb{N}}$:\\[0.1em] $c_i(x)$ in Def.~\ref{def:dlp}} & $\mathsf{DLP}$-$\mathsf{fixed}$ & $\mathsf{CC/QQ}$ & Yes & Yes \\[0.3em]
\hline
$c_{(a,p)}(x)$ in Def.~\ref{def:dlp1} & $\mathsf{DLP}$ & $\mathsf{CC/QQ}$ & Yes & No \\[0.3em]
\hline
$c(x, a, p)$ in Def.~\ref{def:dlp2} & $\mathsf{DLP}$ & $\mathsf{CC/QQ}$ & No & Yes \\ \hline
\end{tabular}
\caption{An overview of the trade-offs regarding the types of separations achievable based on different hardness assumptions for concepts based on the discrete logarithm, providing a fine-grained analysis of the first row of Table~\ref{table:separations}. 
The hardness assumption $\mathsf{DLP}$-$\mathsf{fixed}$ states that there exists an efficiently generatable sequence for which the discrete logarithm is intractable, whereas $\mathsf{DLP}$ is the typical hardness assumption of the discrete logarithm outlined in~\cite{blum:hardness}.
We highlight that to achieve a learning separation for a singleton concept class based on the typical $\mathsf{DLP}$ assumption outlined in~\cite{blum:hardness}, the trade-off is that our construction has to consider non-binary concepts.}
\label{table:assumptions_finedlp}

\begin{tabular}{ c|c|c|c|c}
\diagbox{\textbf{Concepts}}{\textbf{Properties}} & \textbf{Hardness assumption} & \underline{\textbf{Separation}} & \textbf{Binary?} & \textbf{Singleton?} \\[0.3em]
\hline
\makecell{Fix sequences $\{N_i\}_{i \in \mathbb{N}}$:\\[0.2em] $c_j(x)$ in Def.~\ref{def:cube-root}} & $\mathsf{DCRA}$-$\mathsf{fixed}$ & $\mathsf{CC/QC}$ & Yes & Yes \\[0.3em]
\hline
$c_{N,j}(x)$ in Def.~\ref{def:cube-root1} & $\mathsf{DCRA}$ & $\mathsf{CC/QC}$ & Yes & No \\[0.3em]
\hline
$c(N, x)$ in Def.~\ref{def:cube-root2} & $\mathsf{DCRA}$ & $\mathsf{CC/QQ}$ & Yes & Yes 
\end{tabular}
\caption{An overview of the trade-offs regarding the types of separations achievable based on different hardness assumptions for concepts based on the discrete cube root, providing a fine-grained analysis of the second row of Table~\ref{table:separations}. 
The hardness assumption $\mathsf{DCRA}$-$\mathsf{fixed}$ states that there exists an efficiently generatable sequence for which the discrete cube root is intractable, whereas $\mathsf{DCRA}$ is the typical hardness assumption of the discrete cube root outlined in~\cite{kearns:clt}.
We highlight that to achieve a learning separation for a singleton concept class based on the typical $\mathsf{DCRA}$ assumption outlined in~\cite{blum:hardness}, the trade-off is that we achieve a $\mathsf{CC/QQ}$ instead of a $\mathsf{CC/QC}$ separation.}
\label{table:assumptions_finedcr}
\end{table}

\clearpage

\subsection{Learning separation based on heuristic hardness:\ case of discrete~logarithm}
\label{subsec:dlp}

In this section, we study learning separations for concept classes based on the discrete logarithm problem, as first studied in~\cite{liu:dlp}.
The authors of~\cite{liu:dlp} introduce a concept class based on the discrete logarithm and they argue that it exhibits a learning separation which would correspond to a $\mathsf{CC/QQ}$ separation in our nomenclature.
Additionally, they demonstrate its efficient learnability using a general-purpose quantum learning algorithm often referred to as a \textit{quantum kernel method}.

\begin{definition}[Discrete logarithm concept class]
\label{def:dlp}
Let $\{p_n\}_{n\in \mathbb{N}}$ and $\{a_n\}_{n\in \mathbb{N}}$ be fixed sequences of $n$-bit prime numbers $p_n$ and generators $a_n$ of $\mathbb{Z}_{p_n}^*$ (i.e., the multiplicative group of integers modulo $p_n$). 
We define the \emph{discrete logarithm concept class} as $\mathcal{C}^{\mathrm{DL}}_n = \{c_{i}\}_{i \in \mathbb{Z}^*_{p_n}}$, where 
\begin{align}
\label{eq:c_dlp}
     c_i(x) = \begin{cases}1, & \text{if }\log_{(a_n,p_n)} x \in [i, i + \frac{p_n-3}{2}]\footnote{$\log_{(a,p)} x$ denotes the discrete logarithm modulo $p$ of $x$ with respect to the generator~$a$. That is, the discrete logarithm $\log_{(a, p)} x$ is the smallest positive integer $\ell$ such that $a^\ell \equiv x \mod p.$},\\ 0, & \text{otherwise.} \end{cases}
\end{align}
\end{definition}

\begin{remark}[Efficient data generation]
To see why the examples are efficiently generatable for the discrete logarithm class, first note that the examples are of the form 
\begin{align}
\label{eq:example_dlp}
\left(x, c_i(x) \right) = \left(a^y, f_i(y)\right),
\end{align}
where $y \in \{1, \dots, p-1\}$ is the unique integer such that $x \equiv a^y \mod p$, and we let
\begin{align}
f_i(y) = \begin{cases}1, & \text{if }y \in [i, i + \frac{p-3}{2}],\\ 0, & \text{otherwise.} \end{cases}
\end{align}
Secondly, note that $y \mapsto a^y \mod p$ is a bijection from $\{1, \dots, p-1\}$ to $\mathbb{Z}_p^*$, which implies that sampling $x \in \mathbb{Z}_p^*$ uniformly at random is equivalent to sampling $y \in \{0, \dots, p-1\}$ uniformly at random and computing $x = a^y \mod p$.
By combining this observation with Eq.~\eqref{eq:example_dlp}, one finds that one can efficiently generate examples of the discrete logarithm concept $c_i$ under the uniform distribution over $\mathbb{Z}_p^*$ by sampling $y \in \{1, \dots, p-1\}$ uniformly at random, and computing $(a^y, f_i(y))$.
\end{remark}

\paragraph{Hardness assumptions}
While~\cite{liu:dlp} suggests a separation for their concept class by leveraging the canonical hardness assumption of the discrete logarithm as outlined in~\cite{blum:hardness}, we note that we were unable to obtain a learning separation for the concept class in Definition~\ref{def:dlp} based on the hardness assumption in~\cite{blum:hardness} by following the proof ideas in~\cite{liu:dlp}.
Nonetheless, we will discuss two ways to address these shortcomings in the proof of separation.
First, we propose introducing a new and stronger hardness assumption that can be used to obtain a full proof of a learning separation, while ensuring compatibility with quantum kernel methods (since we did not change the concepts when modifying the definition of the concept class).
Next, we will explore an alternative approach to achieve a separation by changing the the concepts.
This change will allows us to leverage the standard hardness of the discrete logarithm outlined in~\cite{blum:hardness} to obtain a learning separation, albeit without direct compatibility with quantum kernel methods.

Our new hardness assumption (which we denote as $\mathsf{DLP}$-fixed) is that there exists a sequence of primes $\{p_n\}_{n\in \mathbb{N}}$ and generators $\{a_n\}_{n \in \mathbb{N}}$ for which the discrete logarithm is classically intractable.
The modified result from~\cite{liu:dlp} is summarized in the following theorem.

\begin{theorem}[modification of~\cite{liu:dlp}]
$L_{\mathrm{DLP}} = \big( \{\mathcal{C}_n^{\mathrm{DLP}}\}_{n \in \mathbb{N}}, \{\mathcal{D}^U_n\}_{n \in \mathbb{N}})$ exhibits a $\mathsf{CC/QQ}$ separation assuming the intractability of the discrete logarithm problem for the fixed sequences of primes $\{p_n\}_{n \in \mathbb{N}}$ and generators $\{a_n\}_{n \in \mathbb{N}}$, where $\mathcal{D}^U_n$ denotes the uniform distribution over $\mathbb{Z}_p^*$.
\end{theorem}
\begin{remark}[Heuristic hardness]
A worst-case hardness assumption is sufficient for classical non-learnability, since the discrete logarithm admits a worst-to-average case reduction~\cite{blum:hardness}.
Note that this worst-to-average case reduction holds for both the canonical hardness assumption as outlined in~\cite{blum:hardness} as well as our newly introduced and stronger $\mathsf{DLP}$-fixed assumption.
\end{remark}

However, we note that the $\mathsf{DLP}$-fixed assumption is stronger than the typical one concerning the discrete logarithm, as outlined in~\cite{blum:hardness}. 
In particular, the hardness assumption outlined in~\cite{blum:hardness} states that there does not exist a polynomial-time classical algorithm that can compute the discrete logarithm for most primes and generators, not that there does not exist a polynomial time algorithm for any fixed sequence of primes and generators.
It would thus be preferable to work with the weakest assumption and construct learning separations that leverage the canonical hardness of the discrete logarithm outlined in~\cite{blum:hardness}.
We show that this is indeed possible and that one can relax the requirement for the existence of sequences of primes and generators for which the discrete logarithm is classically intractable.
Specifically, one can define a different yet closely related concept class, that uses a different input space and underlying set of functions, and show that it exhibits a $\mathsf{CC/QQ}$ separation assuming the canonical hardness of the discrete logarithm outlined in~\cite{blum:hardness}.
\begin{restatable}{definition}{dlponedef}
\label{def:dlp1}
    We define the concept class
    \[
    \hat{\mathcal{C}}^{\mathrm{DL}}_n = \{c_{(a, p)} \mid p \text{ an $n$-bit prime and }a \in \{1, \dots, p-1\}\},
    \]
    where
    \[
        c_{(a, p)}(b, j, x) = \begin{cases}\mathds{1}_{[0, (p-1)/2]}(\log_{(a, p)}(x)) \quad \text{if }b = 0 \\ \mathrm{bin}((a, p), j) \quad \text{if }b=1.  \end{cases}
    \]
    for $b \in \{0,1\}$, $j \in [2n]=\{1, \dots, 2n\}$ and $x \in \{0,1\}^n$.
\end{restatable}

This new concept class essentially involves a mechanism for leaking information about the prime and generator forming the base of the logarithm of the concept in the data.
We delve deeper into this construction in Appendix~\ref{appendix:dlp1} (which may be of independent interest in other contexts) where we provide a formal proof of the following theorem.

\begin{restatable}{theorem}{dlpone}
\label{thm:dlp1}
    $\hat{L}_{\mathrm{DLP}} = \left(\{\hat{\mathcal{C}}^{\mathrm{DL}}_n\}_{n \in \mathbb{N}}, \{\mathcal{D}^U_n\}_{n \in \mathbb{N}} \right)$ exhibits a $\mathsf{CC/QQ}$ separation assuming the intractability of the discrete logarithm problem as outlined in~\cite{blum:hardness}, where $\mathcal{D}_n^{U}$ denotes the uniform distribution over $\mathcal{X}_n = \{0,1\} \times [2n] \times \{0,1\}^n$.
\end{restatable}

\paragraph{Singleton concept class separations}
Another question we would like to address is whether separations are possible for singleton concept classes.
As discussed in Section~\ref{sec:background}, this is crucial in determining whether a problem is a genuine learning problem or merely a computational problem in disguise.
Firstly, we observe that a learning separation for the singleton concept class $\mathcal{C}'_n = \{c\}$, for any choice of $c \in \mathcal{C}^{\mathrm{DL}}_n$ from Definition~\ref{def:dlp}, can be achieved if one assumes our stronger $\mathsf{DLP}$-fixed hardness assumption.
However, one might wonder whether it is also possible to get a separation for a singleton concept class assuming the canonical hardness assumption of~\cite{blum:hardness}.
It turns out that this is indeed possible, though it again requires one to define a different yet closely related singleton concept class that utilizes a similar construction as those in Definition~\ref{def:dlp1}.
\begin{restatable}{definition}{dlptwodef}
\label{def:dlp2}
    We define the concept class $\hat{\mathcal{C}}^{\mathrm{DL}}_n = \{c_n\}$, where
    \[
        c_{n}(a, p, x) = \log_{(a, p)}(x), \quad\text{for $a, p, x \in \{0,1\}^n$.}
    \]
\end{restatable}

\begin{restatable}{theorem}{dlptwo}
    $\hat{L}_{\mathrm{DLP}} = \left(\{\hat{\mathcal{C}}^{\mathrm{DL}}_n\}_{n \in \mathbb{N}}, \{\mathcal{D}^U_n\}_{n \in \mathbb{N}} \right)$ exhibits a $\mathsf{CC/QQ}$ separation assuming the intractability of the discrete logarithm problem as outlined in~\cite{blum:hardness}, where $\mathcal{D}_n^{U}$ denotes the uniform distribution over $\mathcal{X}_n = \{(a, p, x) \mid p\text{ an $n$-bit prime}, a \in \{0,\dots, p-1\}\text{ and } x \in \mathbb{Z}^*_p\}$.
\end{restatable}

A full proof of the above theorem can be found in Appendix~\ref{appendix:dlp2}.
We note that in our construction the learning separation for a singleton concept class assuming the the canonical hardness assumption of~\cite{blum:hardness} comes at the expense of making the concepts have non-binary labels.
We thus notice a trade-off where stronger assumptions are required for separations for binary singleton concept classes, while the difficulty assumption can be relaxed if we consider non-binary concepts.

\subsection{Learning separation based on obfuscation: the case of discrete cube root}
\label{subsec:3root}

In this section, we study learning separations for concept classes based on the discrete cube root, as first studied in~\cite{kearns:clt, servedio:q_c_learnability}. 
In short, in the discrete cube root problem, one is given some moduli $N$ and some input $x\in \mathbb{Z}_N^*$, and one is tasked with finding $y\in \mathbb{Z}_N^*$ such that $y^3 \equiv x \mod N$.
In contrast to the $\mathsf{CC/QQ}$ separations based on the discrete logarithm discussed in the previous section, we will show that the concept classes based on the discrete cube root exhibit a $\mathsf{CC/QC}$ separation.
In particular, one can construct a classically efficient hypothesis class for the concept classes based on the discrete cube root.
Nonetheless, we will show that one still needs a quantum learning algorithm to select the correct hypothesis from the classically efficient hypothesis class.

The authors of~\cite{kearns:crypto} introduce a concept class based on the
discrete cube root and show that it not efficiently classically learnable. 
Additionally, in~\cite{servedio:q_c_learnability} the authors argue that this concept class is efficiently learnable by a quantum learner, and that they thus exhibit a learning separation. 
However, upon closer inspection it turns out that their proof relies on being able to allow the input distribution to depend on the concept you are trying to learn.
In essence, changing the input distribution is used to reveal the moduli $N$ that is being used for a specific concept.
Note that this does not fit within the definition of the PAC learning framework in Section~\ref{sec:background}, where the input distribution has to be the same for all concepts in the concept class. 
One way to make the definition from~\cite{kearns:crypto} comply with our definition of the PAC learning framework without changing the underlying functions, is to fix the moduli for every instance size.

\begin{definition}[Cube root concept class]
\label{def:cube-root}
Let $\{N_i\}_{i \in \mathbb{N}}$ be a sequence of $i$-bit integers $N_i = pq$\footnote{Throughout the paper, the integers $N_i$ are known to the learner beforehand but $p$ and $q$ are not.}, where $p$ and $q$ are two $\lfloor i/2\rfloor$-bit primes such that $\gcd\big(3, (p-1)(q-1)\big) = 1$\footnote{By requiring that $p$ and $q$ satisfy $\gcd\big(3, (p-1)(q-1)\big) = 1$, we ensure that $f_N^{-1}$ exists.
}. We define the \emph{cube root concept class} as $\mathcal{C}_i^{\mathrm{DCR}} = \{c_{j}\}_{j \in [i]}$, with 
\[
c_j(x) = \mathrm{bin}(f_{N_i}^{-1}(x), j),
\]
where $\mathrm{bin}(y, j)$ denotes the $j$th bit of the binary representation of $y$, and the function $f_N^{-1}$ is the inverse of $f_N(x) = x^3 \mod N$ defined on $\mathbb{Z}^*_N$ (i.e., the multiplicative group of integers modulo $N$).
\end{definition}
\begin{remark}[Efficient data generation]
To see why the examples are efficiently generatable for the cube root concept class, first note that the examples are of the form 
\begin{align}
\label{eq:examples_3root}
(x, c_j(x)) = \left(y^3, \mathrm{bin}(y, i)\right),
\end{align}
where $y \in \mathbb{Z}_N^*$ is the unique element such that $x \equiv y^3 \mod N$.
Secondly, note that $f_N(x) = x^3 \mod N$ is a bijection from $\mathbb{Z}_N^*$ to itself, which implies that sampling $x \in \mathbb{Z}_N^*$ uniformly at random is equivalent to sampling $y \in \mathbb{Z}_N^*$ uniformly at random and computing $x = y^3 \mod N$.
By combining this observation with Eq.~\eqref{eq:examples_3root}, one finds that one can efficiently generate examples of the cube root concept $c_j$ under the uniform distribution over $\mathbb{Z}_N^*$ by first sampling $y \in \mathbb{Z}_N^*$ uniformly at random, and then computing $y^3 \mod N$ together with the $j$th bit of the binary representation of $y$.
\end{remark}

\paragraph{Classically efficient hypotheses}
As already mentioned, there is a significant distinction between the learning separations for the discrete logarithm concept class and the cube root concept class that is worth highlighting: the latter is quantumly learnable using a \textit{classically} evaluatable hypothesis class.
To see why this is the case, it is important to note that $f_N^{-1}$ is of the form 
\begin{align}
\label{eq:d}
f_N^{-1}(y) = y^{d^*} \mod N,
\end{align}
for some $d^*$ that only depends on $N$\footnote{In cryptographic terms, $d^*$ is the private decryption key corresponding to the public encryption key $e=3$ and public modulus $N$ in the RSA cryptosystem.}.
The function $f_N^{-1}$ is a type of ``trap-door function'' in that if one is also given $d^*$, then computing $f_N^{-1}$ suddenly becomes classically tractable.
In other words, there exist polynomially-sized Boolean circuits which evaluate this function, whereas for the discrete logarithm we do not know whether such circuits exist.
In this example we thus see the relevance of how the concepts are specified.
The specifications ``$f^{-1}_{N}$ where $f(x) = x^3$'' and ``$f_N^{-1} = x^{d^*}$'' refer to the same functions, yet computing them is in one case classically tractable, and in the other case it is classically intractable (under the DCRA).
The ideas of concealing (easy) functions in difficult descriptions is reminiscent of the term ``obfuscation'' in computer science, and we will use this term in this context as well.
Specifically, we say that the specification "$f^{-1}_{N}$ where $f(x) = x^3$" is an obfuscation of the specification "$f_N^{-1} = x^{d^*}$" in the sense that both represent the same function, but one is not only much harder to understand, but in this case also harder to evaluate.
Using the terminology of Section~\ref{subsec:complexity}, this establishes that the problem of evaluating the concepts actually lies inside $\mathsf{P/poly}$ (where the advice string -- i.e., $d^*$ -- is used to ``de-obfuscate'' the function).

With regards to quantum learnability, in~\cite{servedio:q_c_learnability} the authors note that using Shor's algorithm a quantum learning algorithm can efficiently compute $d^*$ following the standard attack on the RSA cryptosystem.
The cube root concept class is thus quantumly learnable using the classically evaluatable hypothesis class 
\begin{align}
\left\{f_{d,i}(x) = \mathrm{bin}(x^d\text{ mod }N, i)\text{ }\big|\text{ }d\in [N],\text{ }i \in [n]\right\},
\end{align}

Another feature of the cube root concept class which warrants a comment is that even though computing $d^*$ does not require access to the example oracle (recall that $N$ is known  beforehand), we still have to learn the bit of $x^{d^*}$ that is generating the examples, which does require access to the example oracle (i.e., it requires data).

\paragraph{Hardness assumptions}
Much like in the case of the discrete logarithm based separations discussed in Section~\ref{subsec:dlp}, while~\cite{servedio:q_c_learnability} suggests a separation for their concept class by leveraging the canonical hardness assumption of the discrete cube root as outlined in~\cite{kearns:clt}, we note that we were unable to obtain a learning separation for the concept class in Definition~\ref{def:cube-root} based
on the hardness assumption in~\cite{kearns:clt}by following the proof ideas in~\cite{servedio:q_c_learnability}.
Nonetheless, we will discuss two ways of obtaining a learning separation for the modified concept class.
First, we propose introducing a new and stronger hardness assumption that adequately supports a proof of separation.
Next, we will explore an alternative approach to achieve a learning separation by modifying the underlying functions, allowing us to leverage the canonical hardness of the discrete cube root outlined in~\cite{kearns:clt}.

As discussed above, one approach to achieve a $\mathsf{CC/QC}$ separation for the modified concept class defined in~\ref{def:cube-root} involves relying on a stronger hardness assumption.
Specifically, this assumption (which we denote as $\mathsf{DCR}$-fixed) is that there exists a sequence of moduli $\{N_i\}_{i\in \mathbb{N}}$ for which the discrete cube root is classically intractable.
We summarize the results regarding the separation of the modified cube root concept class in the following theorem.

\begin{theorem}[\cite{servedio:q_c_learnability, kearns:clt}]
$L_{\mathrm{DCR}} = \big( \{\mathcal{C}_i^{\mathrm{DCR}}\}_{i \in \mathbb{N}}, \{\mathcal{D}^U_i\}_{i \in \mathbb{N}})$ exhibits a $\mathsf{CC/QC}$ separation assuming the intractability of the discrete cube root for the fixed sequence of moduli $\{N_i\}_{i \in \mathbb{N}}$, where $\mathcal{D}^U_i$ denotes the uniform distribution over $\mathbb{Z}_{N_i}^*$.
\end{theorem}
\begin{remark}[Heuristic hardness]
A worst-case hardness assumption is sufficient for classical non-learnability, as the discrete cube root admits a worst-to-average case reduction~\cite{alexi:rsa, goldwasser:rsa}.
Note that this worst-to-average case reduction holds for both the canonical hardness assumption as outlined in~\cite{kearns:clt} as well as our newly introduced and weaker $\mathsf{DCR}$-fixed assumption.
\end{remark}

However, we again note that the $\mathsf{DCR}$-fixed assumption is stronger than the typical one concerning the discrete cube root, as outlined in~\cite{kearns:clt}.
In particular, the hardness assumption outlined in~\cite{kearns:clt} states that there does not exist a polynomial-time classical algorithm that can compute the discrete cube root for most moduli, not that there does not exist a polynomial-time algorithm for any fixed sequence of moduli.
This prompts the question of whether learning separations can also be achieved by leveraging the canonical hardness of the discrete cube root outlined in~\cite{kearns:clt}. 
We show that this is indeed possible and that one can relax the requirement for the existence of sequences of moduli for which the discrete cube root is classically intractable. 
Specifically, one can define a different yet closely related concept class, that uses a different input space and underlying set of functions, and show that it exhibits a $\mathsf{CC/QC}$ separation assuming the canonical hardness of the discrete logarithm outlined in~\cite{kearns:clt}.

\begin{restatable}{definition}{dcronedef}
\label{def:cube-root1}
    We define the concept class $\hat{\mathcal{C}}^{\mathrm{DCR}}_n = \{c_{(N, i)} \mid N = pq \text{ as in Definition~\ref{def:cube-root}}, \text{ }i \in [n]\}$, where
    \[
        c_{(N, i)}(b, j, x) = \begin{cases}\mathrm{bin}(f^{-1}_{N}(x), i) \quad &\text{if }b = 0 \\ \mathrm{bin}(N, j) \quad &\text{if }b=1.  \end{cases}
    \]
    for $b \in \{0,1\}$, $j \in [n]=\{1, \dots, n\}$ and $x \in \{0,1\}^n$.
\end{restatable}

This new concept class essentially involves a mechanism for leaking information about the moduli used by the concept in the data. 
We defer a formal proof of the theorem below to Appendix~\ref{appendix:dcr1}.

\begin{restatable}{theorem}{dcrone}
    $\hat{L}_{\mathrm{DCR}} = \left(\{\hat{\mathcal{C}}^{\mathrm{DCR}}_n\}_{n \in \mathbb{N}}, \{\mathcal{D}^U_n\}_{n \in \mathbb{N}} \right)$ exhibits a $\mathsf{CC/QC}$ separation assuming the intractability of the discrete cube root as outlined in~\cite{kearns:clt}, where $\mathcal{D}_n^{U}$ denotes the uniform distribution over $\mathcal{X}_n = \{0,1\} \times [n] \times \{0,1\}^n$.
\end{restatable}

\paragraph{Singleton concept class separations}

Another question we would like to address is whether separations are possible for singleton concept classes.
As discussed in Section~\ref{sec:background}, this is crucial in determining whether a problem is a genuine learning problem or merely a computational problem in disguise.
Firstly, we observe that a learning separation for the singleton concept class $\mathcal{C}_i = \{c_i\}$, where $c_i \in \mathcal{C}^{\mathrm{DCR}}_i$ corresponds to the least significant bit, can be achieved if one assumes our new, yet stronger $\mathsf{DCR}$-fixed hardness assumptions.
However, one might wonder whether it is also possible to get a separation for a singleton concept class assuming the canonical hardness assumption of~\cite{kearns:clt}.
It turns out that this is indeed the case, though it requires one to define a different yet closely related singleton concept class that utilizes a similar construction to Definition~\ref{def:cube-root1}.

\begin{restatable}{definition}{dcrtwodef}
\label{def:cube-root2}
    We define the concept class $\hat{\mathcal{C}}^{\mathrm{DCR}}_n = \{c_n\}$, where
    \[
        c_n(N, x) = \mathrm{bin}(f^{-1}_{N}(x), n) \quad\text{for $N, x \in \{0,1\}^n$.}
    \]
\end{restatable}

\begin{restatable}{theorem}{dcrtwo}
    $\hat{L}_{\mathrm{DCR}} = \left(\{\hat{\mathcal{C}}^{\mathrm{DCR}}_n\}_{n \in \mathbb{N}}, \{\mathcal{D}^U_n\}_{n \in \mathbb{N}} \right)$ exhibits a $\mathsf{CC/QQ}$ separation assuming the intractability of the discrete cube root as outlined in~\cite{kearns:clt}, where $\mathcal{D}_n^{U}$ denotes the uniform distribution over $\mathcal{X}_n = \{(N, x) \mid N=pq\text{ as in Definition~\ref{def:cube-root}} \text{ and } x \in \mathbb{Z}^*_N\}$.
\end{restatable}

A formal proof of the above theorem can be found in Appendix~\ref{appendix:dcr2}.
We note that in our construction the learning separation for a singleton concept class assuming the the canonical hardness assumption of~\cite{blum:hardness} comes at the expense of having to change the type of learning separation that is achieved from a $\mathsf{CC/QC}$ separation to a $\mathsf{CC/QQ}$.
Therefore, we notice again tradeoff where stronger assumptions are required for $\mathsf{CC/QC}$ separations for binary singleton concept classes, while the difficulty assumption can be relaxed for $\mathsf{CC/QQ}$ separations for binary singleton concept classes

\subsection{Learning separation with efficiently evaluatable concepts}
\label{subsec:new_sep}

In this section we establish a learning separation (contingent on a plausible though relatively unexplored hardness assumption) where the concepts do not just admit polynomial-sized Boolean circuits, but are also given in a representation which is efficiently evaluatable on a classical computer.
For this concept class, the hardness of learning them cannot stem from the hardness of \textit{evaluating} the concepts, and it thus lies in \textit{identifying} which specific concept is generating the examples.
To the best of our knowledge, no such separation was given in the literature before.
The concept class that satisfies all of the above is the modular exponentiation concept class defined as follows.

\begin{definition}
\label{def:modexp}
We define the \emph{modular exponentiation concept class} as
\[
\mathcal{C}^{\mathrm{modexp}}_n = \{c_{(N, d)} \mid N=pq\text{ an $n$-bit $2^c$-integer as in Definition~\ref{def:2c} }, 0 \leq d \leq (p-1)(q-1)\},
\]
where
\[
    c_{(N, d)}(b, x) = \begin{cases}x^d \mod N  &\text{if }b = 0 \\ N &\text{if }b=1.  \end{cases}
\]
for $b \in \{0,1\}$ and $x \in \{0,1\}^n$.
\end{definition}

\begin{remark}
The concepts are not binary-valued, and it is an open question whether and how the separation can be translated to also hold for binary-valued concepts.
\end{remark}

\begin{definition}[$2^c$-integer]
\label{def:2c}
    An $n$-bit integer $N = pq$ is a $2^c$-integer if $p$ and $q$ are two $\lfloor n/2 \rfloor$-bit primes such that $\gcd\big(3, (p-1)(q-1)\big) = 1$ and:
    \begin{itemize}
    \item[(i)] There exists a constant $c$ (i.e., independent of $n$) such that $2^c \nmid (p-1)(q-1)$.
    \item[(ii)] There exists a constant $c'$ (i.e., independent of $n$) such that $\gcd(p-1, q-1) = 2^{c'}$.
    \end{itemize}
\end{definition}

We summarize the learning separation of the modular exponentiation concepts in Theorem~\ref{thm:modexp}, and defer the proof to Appendix~\ref{appendix:modexp}, and we discuss various aspects of this learning separation below.

\begin{restatable}{theorem}{modexp}
\label{thm:modexp}
$L_{\mathrm{modexp}} = \left(\{\mathcal{C}^{\mathrm{modexp}}_n\}_{n \in \mathbb{N}}, \{\mathcal{D}^U_n\}_{n \in \mathbb{N}} \right)$ exhibits a $\mathsf{CC/QC}$ separation assuming the intractability of the discrete cube root as stated in~\cite{kearns:clt} when restricted to $2^c$-integer moduli\footnote{Here we mean the intractability assumption that states that no classical algorithm can efficiently compute the discrete cube root $f_N^{-1}(x)$ for a $\frac{1}{2} + \frac{1}{\mathrm{poly}(n)}$ fraction of all $2^c$-integer moduli $N$ and inputs $x \in \mathbb{Z}_N^*$.}, where $\mathcal{D}_n^{U}$ denotes the uniform distribution over $\mathcal{X}_n = \{0,1\} \times \{0,1\}^n$.
\end{restatable}

\paragraph{Quantum learnability}
To show that the modular exponentiation concept class is quantumly learnable, we use a combination of the quantum algorithm for order-finding and the quantum algorithm for the discrete logarithm~\cite{shor:factoring}.
The key observation is that an example $(x, x^d\text{ mod }N)$ specifies a congruence relation $d \equiv a\text{ mod }r$, where $r$ denotes the multiplicative order of $x \in \mathbb{Z}^*_N$, and $a$ denotes the discrete logarithm of $x^d$ in the subgroup generated by $x$ (i.e., the smallest positive integer $\ell$ such that $x^\ell \equiv x^d\mod N$).
Next, using the fact that $N$ is a $2^c$-number, we show that a polynomial number of these congruences suffices to recover $d$ with high probability.

\paragraph{Hardness assumption}
We show that the above concept class is not classically learnable assuming the intractability of the discrete cube root as stated in~\cite{kearns:clt} restricted to $2^c$-integer moduli{\color{blue}\footnotemark[7]}.
We will refer to this assumption as the $2^c$-discrete cube root assumption ($2^c$-DCRA).
To see this, note that the modular exponentiation concept class contains the cube root function $f_N^{-1}$ discussed in Section~\ref{subsec:3root} (though this time it is not ``obfuscated'').
Moreover, using the construction also outlined in Section~\ref{subsec:3root} we can efficiently generate examples $(y, f_N^{-1}(y))$, for $y \in \mathbb{Z}_N^*$ uniformly at random.
If we put these examples into an efficient classical learning algorithm for the modular exponentiation concept class, it would with high probability identify a classically efficiently evaluatable hypothesis that agrees with $f_N^{-1}$ on a $1 - \frac{1}{\mathrm{poly}(n)}$ fraction of inputs.
Similar to Section~\ref{subsec:3root}, by the worst-case to average-case reduction of~\cite{alexi:rsa, goldwasser:rsa} this then violates the $2^c$-DCRA.

We note that by imposing that $N$ is a $2^c$-integer might cause the $2^c$-DCRA to no longer hold, since there could be an efficient classical algorithm for these specific $2^c$-integer moduli. 
However, since $2^c$-integers  are generally not considered to be unsecure or ``weak'' moduli for the RSA cryptosystem, and since recently factored RSA numbers\footnote{\url{https://en.wikipedia.org/wiki/RSA_numbers}} are all essentially $2^c$-integers, it is plausible that the DCRA still holds when restricted to $2^c$-integers (see Appendix~\ref{appendix:moduli} for more details).

\paragraph{Conclusion}
The modular exponentiation concept class thus exhibits a $\mathsf{CC/QC}$ separation (assuming the $2^c$-DCRA hold), where the concepts are classically efficiently evaluatable.
Since the concepts are classically efficiently evaluatable, one could argue that the classical hardness of learning lies in \textit{identifying} rather than \textit{evaluating} a hypothesis that is close to the concept generating the examples.
We remark that for the modular exponentation concept class, it is not possible to restrict the concept class and obtain a similar learning separation where a quantum learner does not require any data (i.e., similar to Observation~\ref{lemma:trivial_sep}).
In fact, since the concepts are efficiently evaluatable classically, any polynomially-sized subset of concepts is classically learnable since a classical learning algorithm can do a brute-force search to find the concept that best matches the data.

In the next section, we present an example of a separation in the setting where the learner is constrained to only
output hypotheses from a fixed hypothesis class.
Since the learner is not required to evaluate the concepts on unseen examples, it can be argued that in this case the classical hardness also lies in identifying 
  rather than evaluating the concept generating the examples.
  
\subsection{Learning separation with a fixed hypothesis class}
\label{subsec:elgamal}

In this section we establish a separation in the setting where the learner is constrained to only output hypotheses from a fixed hypothesis class.
Recall that in this setting the learner is not required to be able to evaluate the concepts, so the hardness of learning must stem from the hardness of identifying the hypothesis that is close to the concept generating the data. 
The main differences compared to the modular exponentiation concept class are that the concepts discussed in this section are binary-valued and that it is unknown whether they exhibit a separation in the setting where the learner is free to output arbitrary hypotheses.
The concept class we discuss in this section is defined below, and it is a modification of the cube root concept class from Definition~\ref{def:cube-root}.

\begin{definition}[Cube root identification concept class]
\label{def:cube-root-id}
We define the \emph{cube root identification} concept class as 
\[
    \mathcal{C}^{\mathrm{DCRI}}_n = \Big\{c_{(N, m)} \mid N = pq \text{, $p$ and $q$ two $\lfloor n/2\rfloor$-bit primes s.t.\ $\gcd\big(3, (p-1)(q-1)\big) = 1$}, \text{ }m \in \mathbb{Z}_N^*\Big\},
\]
where
    \[
        c_{(N, m)}(b, x) = \begin{cases}\mathrm{bin}(m^3\text{ mod } N,\text{  }\mathrm{int}(x_1:\dots:x_{\lfloor \log n \rfloor})) \quad &\text{if }b = 0 \\ \mathrm{bin}(N, \mathrm{int}(x_1:\dots:x_{\lfloor \log n \rfloor})) \quad &\text{if }b=1.  \end{cases}
    \]
for $b \in \{0,1\}$, $x \in \{0,1\}^n$.
\end{definition}
\begin{remark}
    Here $\mathrm{bin}(y, k)$ denotes the $k$th bit of the binary representation of $y$, and we additionally use $\mathrm{int}(x_1:\dots:x_{\lfloor \log n \rfloor})$ to denote the integer encoded by the first $\lfloor \log n \rfloor$-bits of $x \in \{0,1\}^n$.
\end{remark}

We summarize the learning separation of the cube root identification concepts in Theorem~\ref{thm:modexp}, we defer the proof to Appendix~\ref{appendix:elgamal}, and we discuss various aspects of this learning separation below.
\begin{restatable}{theorem}{elgamal}
\label{thm:elgamal}
$L_{\mathrm{DCRI}} = \left(\{\mathcal{C}^{\mathrm{DCRI}}_n\}_{n \in \mathbb{N}}, \{\mathcal{D}^U_n\}_{n \in \mathbb{N}} \right)$ exhibits a $\mathsf{C}_{\mathcal{H}}/\mathsf{Q}_{\mathcal{H}}$ separation assuming the intractability of the discrete cube root as outlined in~\cite{kearns:clt}, where $\mathcal{D}_n^{U}$ denotes the uniform distribution over $\mathcal{X}_n = \{0,1\} \times \{0,1\}^n$.
\end{restatable}

\paragraph{Quantum learnability}
To establish that the cube root identification concept class  is quantumly proper learnable, we first note that using $\mathcal{O}(\mathrm{poly}(n))$ examples of a concept $c_m$ under the uniform distribution we can with high probability reconstruct the full binary representation of $m^3$.
Since we can also with high probability recover $N$ using the same amount of examples, we can use Shor's algorithm~\cite{shor:factoring} to compute $d$ such that $(m^3)^d \equiv m\text{ mod }N$ (allowing us to correctly identify the concept $c_m$). 
We remark that the quantum learner needs access to the data in order to obtain a full reconstruction of the binary representation of both $m^3 \mod N$ as well as the modulus $N$.

\paragraph{Hardness assumption}
We show that the cube root identification concept class is not classically learnable with a fixed hypothesis class under the \emph{Discrete Cube Root Assumption} (DCRA) discussed in Section~\ref{subsec:3root}.
To show that the existence of an efficient classical learner violates the DCRA, we let $e \in \mathbb{Z}^*_N$ and show we how an efficient classical learner can efficiently compute $m = f_N^{-1}(e)$.
First, we generate examples $(x, \mathrm{bin}(e,k))$ and $(x, \mathrm{bin}(N, k))$, where $k = \mathrm{int}(x_1:\dots:x_{\lfloor \log n \rfloor})$.
When plugging these examples into an efficient classical learner it will with high probability identify an $m'$ such that $(m')^3 \equiv m \mod N$.
Since $x \mapsto x^3 \mod N$ is a bijection on $\mathbb{Z}^*_N$ we find that $m = m'$, and conclude that an efficient classical learner can efficiently compute the discrete cube root $f_N^{-1}(e)$.

\paragraph{Conclusion} The cube root identification concept class exhibits a separation in the setting where the learner is constrained to only output hypotheses from a fixed hypothesis class.
In fact, this is a separation in the so-called \textit{proper} PAC framework, since the hypothesis class is the same as the concept class.
Since in this setting it is not required to evaluate the concepts on unseen examples, the classical hardness has to lie in \textit{identifying} rather than \textit{evaluating} the concept generating the examples.
Note that for the  cube root identification concept class it is not possible to obtain a learning separation for a singleton concept class where a learner does not require any data (see Observation~\ref{lemma:trivial_sep}).

\setcounter{footnote}{0} 
\section{Learning separations without efficient data generation}
\label{sec:no_eff_data}

In the quantum machine learning community there is an often-mentioned conjecture that quantum machine learning is most likely to have its advantages for data that is generated by a ``genuine quantum process''\footnote{Recently, there have been notable developments that have yielded contrasting conclusions. For instance, in~\cite{huang:science}, surprisingly complex physics problems are efficiently learned by classical learners. We will briefly discuss this in Section~\ref{subsec:huang}.}.
We understand this to mean that the concepts generating the data are $\mathsf{BQP}$-complete or perhaps $\mathsf{DQC1}$-complete. 
It is worth noting that if concepts in $\mathsf{BQP}$ or $\mathsf{DQC1}$ that are not in $\mathsf{BPP}$ are already considered a ``genuine quantum process'', then the discrete logarithm concept class discussed in Section~\ref{subsec:dlp} suffices. 
However, we aim to investigate learning separations beyond these concepts, i.e., where the concepts are $\mathsf{BQP}$-complete.

A natural question that arises is, given a family of $\mathsf{BQP}$-complete concepts, what additional assumptions are sufficient to prove that these concepts exhibit a learning separation?
In Section~\ref{sec:eff_data}, we discussed proofs of learning separations that were predicated on the data being efficiently generatable by a classical device.
However, since there is no reason to believe that a family of $\mathsf{BQP}$-complete concepts allow for efficient data generation, we will need to adopt a different proof-strategy.

To ensure quantum learnability of a family of $\mathsf{BQP}$-complete concepts $\{\mathcal{C}_n\}_{n \in \mathbb{N}}$, we can simply limit the size of each concept class $\mathcal{C}_n$ to be no more than a polynomial in $n$.
When the size of the concept class is polynomial, a quantum learner can iterate over all concepts and identify the concept that best matches the examples from the oracle.
In more technical terms, a quantum learner can efficiently perform empirical risk minimization through brute-force fitting.
From standard results in learning theory (e.g., Corollary 2.3 in~\cite{shalev:book}), it follows that this method results in a learner that satisfies the conditions of the PAC learning framework.
 
As discussed in Section~\ref{subsec:complexity}, assuming that the concepts are not in $\mathsf{HeurP/poly}$ is sufficient to ensure that the concept class is not classically learnable. 
Intuitively, this is because if the concepts were classically learnable, the examples could be used to construct an advice string that, together with an efficient classical learning algorithm, would put the concepts in $\mathsf{HeurP/poly}$. 
By combining this with our approach to ensure quantum learnability, we can show that if there exists a family of polynomially-sized concept classes consisting of $\mathsf{BQP}$-complete concepts that are not in $\mathsf{HeurP/poly}$, then this family of concept classes exhibits a  $\mathsf{CC/QQ}$ separation. 
Moreover, in Section~\ref{subsec:physical_systems} we discuss how several of these separations can be build around data that is generated by a ``genuine quantum process''.
The following theorem summarizes our findings, and we defer the proof to Appendix~\ref{appendix:seps_no-eff-data}.

\begin{restatable}{theorem}{noeffdata}
\label{thm:seps_no-eff-data}
Consider a family of concept classes $\{\mathcal{C}_n\}_{n \in \mathbb{N}}$ and distributions $\{\mathcal{D}_n\}_{n \in \mathbb{N}}$ such that

\medskip

\noindent \underline{Quantum learnability:}
\begin{itemize}
    \item[(a)] Every $c_n \in \mathcal{C}_n$ can be evaluated on a quantum computer in time $\mathcal{O}\left(\mathrm{poly}(n)\right)$.
    \item[(b)] There exists a polynomial $p$ such that for every $n \in \mathbb{N}$ we have $|\mathcal{C}_n| \leq p(n)$.
\end{itemize}

\smallskip

\noindent \underline{Classical non-learnability:}
\begin{itemize}
    \item[(c)] There exists a family $\{c_n\}_{n \in \mathbb{N}}$, where $c_n \in \mathcal{C}_n$, such that $(\{c_n\}_{n \in \mathbb{N}}, \{\mathcal{D}_n\}_{n \in \mathbb{N}})\not\in \mathsf{HeurP/poly}$.
\end{itemize}

\smallskip

\noindent Then, $L = (\{\mathcal{C}_n\}_{n \in \mathbb{N}}, \{\mathcal{D}_n\}_{n \in \mathbb{N}})$ exhibits a $\mathsf{CC/QQ}$ learning separation.
\end{restatable}

At face value, it may not be clear whether there exist concept classes that satisfy both conditions (a) and (c), since condition (a) puts the concepts in $\mathsf{BQP}$ and it may not be clear how large $\mathsf{HeurP/poly}$ is relative to $\mathsf{BQP}$. 
Notably, it is known that if the discrete logarithm is not in $\mathsf{BPP}$, then it is also not in $\mathsf{HeurBPP}$ under certain distributions. 
Additionally, it is widely believed that a polynomial amount of advice does not significantly improve the computational complexity of solving the discrete logarithm problem~\cite{corrigan:dlp}. 
Hence, it is plausible to imagine the existence of problems $L \in \mathsf{BQP}$ for which there is a distribution $\mathcal{D}$ such that $(L, \mathcal{D}) \not\in \mathsf{HeurP/poly}$. Moreover, it is interesting to observe that if there exists a single $L \in \mathsf{BQP}$ that is not in $\mathsf{HeurP/poly}$ under some distribution, then for every $\mathsf{BQP}$-complete problem there exists a distribution under which it is not in $\mathsf{HeurP/poly}$. 
We summarize this in the lemma below, and we defer the proof to Appendix~\ref{appendix:lemma1}.

\begin{restatable}{lemma}{lemmaone}
\label{lemma:1}
If there exists a $(L, \mathcal{D}) \not\in \mathsf{HeurP/poly}$ with $L \in \mathsf{BQP}$, then for every $L' \in \mathsf{BQP}$-$\mathsf{complete}$\footnote{With respect to many-to-one reductions (as is the case for, e.g., quantum linear system solving~\cite{harrow:qls}).} there exists a family of distributions $\mathcal{D}' = \{\mathcal{D}'_n\}_{n \in \mathbb{N}}$ such that $(L', \mathcal{D}') \not\in \mathsf{HeurP/poly}$.
\end{restatable}

In summary, to obtain a learning separation for data generated by a ``genuine quantum process'', it is sufficient to have a single problem $L \in \mathsf{BQP}$ that lies outside $\mathsf{HeurP/poly}$ under some distribution. 
An example of such a problem is the discrete logarithm. 
However, the resulting distribution under which the $\mathsf{BQP}$-complete problem lies outside of $\mathsf{HeurP/poly}$ is artificial as it comes from explicitly encoding
the discrete logarithm into the learning problem through the reduction to the $\mathsf{BQP}$-complete problem.
Besides the discrete logarithm, little is known about the heuristic hardness of problems in $\mathsf{BQP}$ (especially those that are considered ``genuinely quantum'').
Therefore, the question arises as to what additional properties are required for a $\mathsf{BQP}$-complete problem to lie outside $\mathsf{HeurP/poly}$ under some distribution. 
We show that a worst-case to average-case reduction combined with the assumption that $\mathsf{BQP} \not\subseteq \mathsf{P/poly}$ is sufficient for this purpose. 
While the question of $\mathsf{BQP} \not\subseteq \mathsf{P/poly}$ remains open, we proceed under this assumption based on its implications for cryptography.
Specifically, if $\mathsf{BQP} \subseteq \mathsf{P/poly}$, then problems like the discrete logarithm would be in $\mathsf{P/poly}$, which would break cryptographic systems assumed to be secure\footnote{In cryptography it is common to assume non-uniform adversaries (i.e., with computational resources of $\mathsf{P/poly}$), and even in this case most public-key cryptosystems such as RSA and Diffie-Hellman are still assumed to be safe).}. 
Under the assumption that $\mathsf{BQP} \not\subseteq \mathsf{P/poly}$, the only missing piece is that our problem $L \in \mathsf{BQP}$ that lies outside $\mathsf{P/poly}$ is random self-reducible with respect to some distribution (i.e., it admits a worst-case to average case reduction as discussed in Section~\ref{subsec:complexity}). 
We summarize these findings in the lemma below, and we defer the proof to Appendix~\ref{appendix:lemma2}.

\begin{restatable}{lemma}{lemmatwo}
\label{lemma:2}
If $L \not\in \mathsf{P/poly}$ and $L$ is polynomially random self-reducible with respect to some distribution $\mathcal{D}$, then $(L, \mathcal{D}) \not\in \mathsf{HeurP/poly}$.
\end{restatable}

By combining Lemma~\ref{lemma:1} with Lemma~\ref{lemma:2}, we obtain a set of assumptions that result in provable learning separations for data that could be generated by a genuine quantum process, as stated in Theorem~\ref{thm:seps_no-eff-data} (see also Section~\ref{subsec:physical_systems}).
These assumptions include the existence of a problem $L \in \mathsf{BQP}$ that is not in $\mathsf{P/poly}$ which is polynomially random self-reducible with respect to some distribution.

\begin{corollary}
\label{cor:seps}
If there exists an $L \in \mathsf{BQP}$ such that $L \not\in \mathsf{P/poly}$ and it is random self-reducible, then every $\mathsf{BQP}$-complete problem gives rise to a $\mathsf{CC/QQ}$ separation.
\end{corollary}

Although establishing such learning separations is not straightforward, the criteria listed in Lemma~\ref{lemma:1}, Lemma~\ref{lemma:2} and Corollary~\ref{cor:seps} suggest some challenges that when addressed lead to provable learning separations.
In particular, they highlight the need for further investigation into the heuristic hardness of problems in $\mathsf{BQP}$ from the perspective of quantum machine learning.

An interesting observation can be made when we consider the constructions underlying Corollary~\ref{cor:seps} and let $L = \mathsf{DLP}$ to be the language corresponding to the discrete logarithm problem. 
Firstly, we note that $\mathsf{DLP} \in \mathsf{BQP}$ due to Shor's algorithm~\cite{shor:factoring}.
Secondly, since $\mathsf{DLP}$ is both random self-reducible~\cite{blum:hardness} and random verifiable (i.e., we can generate examples under the uniform distribution as in Section~\ref{subsec:dlp}), we know that $\mathsf{DLP} \in \mathsf{HeurBPP/samp}$ is equivalent to $\mathsf{DLP} \in \mathsf{BPP}$.
By combining these observations with the constructions underlying Corollary~\ref{cor:seps} we obtain the lemma below, which shows that for every $\mathsf{BQP}$-$\mathsf{complete}$ problem we can construct a distribution $\mathsf{D}_L$ under which it exhibits a $\mathsf{CC}/\mathsf{QQ}$ separation unless $\mathsf{DLP} \in \mathsf{BPP}$ (we defer the proof to Appendix~\ref{appendix:cor2}).

\begin{restatable}{lemma}{cortwo}
\label{cor:seps1}
For every $L \in \mathsf{BQP}$-$\mathsf{complete}$ there exists an efficiently samplable distribution $\mathcal{D}_L$ such that $(L, \mathcal{D}_L) \not\in \mathsf{HeurBPP}/\mathsf{samp}$, unless $\mathsf{DLP} \in \mathsf{BPP}$.
\end{restatable}

\paragraph{Generative modeling}
To the authors it is not entirely clear precisely how strong the assumption that $\mathsf{BQP} \not\subseteq \mathsf{P/poly}$ is.
It is worth noting though that for sampling problems arguably more iron-clad assumptions, such as the non-collapse of the polynomial hierarchy, could potentially lead to analogous conclusions.
In particular, one possibility is to use quantum supremacy arguments~\cite{aaronson:supremacy, bremner:supremacy} to establish learning separations in generative modeling, where the task is to learn a distribution instead of a binary function.
If the distribution to be learned is in $\mathsf{SampBQP}$ (i.e., sampling problems solvable by a polynomial-time quantum algorithm), then for classical non-learnability, the corresponding requirement is that not all $\mathsf{SampBQP}$ problems are in $\mathsf{SampBPP/poly}$ (i.e., sampling problems solvable by a polynomial-time classical algorithm with polynomial-sized advice)\footnote{For a formal definition of these complexity classes we refer to~\cite{aaronson:supremacy}.}. 
This is analogous to the supervised learning case, but for sampling problems we might have further evidence this is unlikely. 
Specifically, as sketched by Aaronson in~\cite{aaronson:blog}, if $\mathsf{SampBQP} \subseteq \mathsf{SampBPP/poly}$, then this could cause the polynomial hierarchy to collapse.
In other words, one could arguably use these arguments to show that a family of distributions is not classically learnable, under the assumption that the polynomial hierarchy does not collapse.

\subsection{Learning separations from physical systems}
\label{subsec:physical_systems}

Many quantum many-body problems are either $\mathsf{BQP}$-complete or $\mathsf{QMA}$-complete when appropriately formalized, making them suitable for defining concepts that are not classically learnable (recall that this also implies a learning separation, since quantum learnability can be ensured by considering polynomially-sized concept classes). 
To be more precise, recall that any problem in $\mathsf{BQP}$ that does not lie in $\mathsf{HeurP/poly}$ with respect to some distribution can be used to construct a distribution under which a hard quantum many-body problem defines a learning problem that is not classically learnable (as shown by Theorem~\ref{thm:seps_no-eff-data} and Lemma~\ref{lemma:1}). 
However, the induced distribution under which the physical system is not classically learnable is artificial, as it is induced by a particular choice of reduction, and there is no evidence that these induced distributions are relevant in practice.

\paragraph{Examples of physical systems}
For concreteness, let us discuss some examples.
There are many physical systems that are in some sense universal for quantum computing, such as the Bose-Hubbard model~\cite{childs:qma}, the antiferromagnetic Heisenberg and antiferromagnetic XY model~\cite{piddock:qma}, the Fermi-Hubbard model~\cite{gorman:electronic}, supersymmetric systems~\cite{cade:susy}, interacting bosons~\cite{wei:qma}, and interacting fermions~\cite{liu:qma}.
In particular, each of these physical systems defines a family of Hamiltonians and, for several of these Hamiltonian families, time-evolution is $\mathsf{BQP}$-complete when appropriately formalized~\cite{haah:ham_sim, childs:hs}.
That is, for several of these universal Hamiltonian families $H(\beta)$, where $\beta$ denote the Hamiltonian parameters, we can define $\mathsf{BQP}$-complete concepts
\[
c_H(\beta, t) = \mathrm{sign}\left(\big|\langle 0^n|e^{iH(\beta)t}Z_1e^{-iH(\beta)t} |0^n\rangle\big|^2 
 - \frac{1}{2}\right),
\]
where $Z_1$ denotes the Pauli-$Z$ operator on the first qubit and identity elsewhere.
Additionally, one could also use $\mathsf{BQP}$-complete problems in high energy physics, such as scattering in scalar quantum field theory~\cite{jordan:bqp_hep}.

As another example, we note that for any of the universal Hamiltonian families the problem of finding the ground state energy is $\mathsf{QMA}$-complete.
That is, for any universal Hamiltonian family $H(\beta)$, where $\beta$ denote the Hamiltonian parameters, we can define $\mathsf{QMA}$-complete concepts
\[
    c_H(\beta) = \mathrm{sign}\left( \mathrm{Tr}\left[H(\beta) \ket{\psi_{H}(\beta)} \right] - \frac{1}{2} \right),
\]
where $\ket{\psi_{H}(\beta)}$ denotes the ground state of $H(\beta)$.
Naturally, one worries that these concepts are too hard to evaluate on a quantum computer, but there are a few workarounds.
Firstly, sometimes there is a natural special case of the problem that is $\mathsf{BQP}$-complete (e.g., the subset of Hamiltonians obtained through a circuit-to-Hamiltonian mapping).
Moreover, more generically it holds that any problem that is $\mathsf{QMA}$-complete has a restriction that is $\mathsf{BQP}$-complete (i.e., take any $\mathsf{BQP}$-complete problem and consider the image of this problem under a many-to-one reduction).
Finally, one could use recent results on the guided local Hamiltonian problem to relax the $\mathsf{QMA}$-complete problems and obtain a $\mathsf{BQP}$-complete problem~\cite{weggemans:guidable, cade:guidable, gharibian:guidable}.

In short, by exploiting a reduction from a problem that is in $\mathsf{BQP}$ which under a given distribution lies outside $\mathsf{HeurP/poly}$ onto a chosen $\mathsf{BQP}$-complete problem (as in Lemma~\ref{lemma:1}), any physical system that is in some sense universal for quantum computing can be used to construct a learning separation.
Nonetheless, since the reduction is implicitely used to construct the distribution under which the physical system becomes not classically learnable, the distributions will be artificial and there is no reason to believe these have any relevance in practice.

\setcounter{footnote}{0} 
\section{Connections to other works on (quantum) learning tasks}
\label{sec:discussion}

In this section we discuss other topics of relevance.
First, in Section~\ref{subsec:huang}, we discuss the implications and limitations of the milestone work of Huang et al.~\cite{huang:science} on establishing learning separations from physical systems.
Next, in Section~\ref{subsec:power_data}, we discuss how having access to data radically enhances what can be efficiently evaluated by discussing the example of evaluating parameterized quantum circuits.
Afterwards, in Section~\ref{subsec:proper_physics}, we discuss how two physically-motivated problems (i.e., Hamiltonian learning, and identifying order parameters for phases of matter) fit in the PAC learning setting where the learner is constrained to output hypotheses from a fixed hypothesis class.

\subsection{Provably efficient machine learning with classical shadows}
\label{subsec:huang}

In the milestone work of Huang et al.~\cite{huang:science}, the authors design classical machine learning methods (in part built around the \emph{classical shadow} paradigm) that can efficiently learn quantum many-body problems.
One of the problems studied in~\cite{huang:science} is that of \textit{predicting ground states of Hamiltonian}.
More precisely, for a family of Hamiltonians $H(x)$ with ground states $\rho_H(x)$, one wants to predict the expectation value of some observable $O$ when measured on $\rho_H(x)$.
That is, one wants to efficiently learn to evaluate the function
\begin{align}
\label{eq:huang_science}
    f_{H, O}(x) = \mathrm{Tr}\left[\rho_H(x) O \right].
\end{align}

One of the main things that~\cite{huang:science} show is that given a polynomial number of data points, one is able to efficiently evaluate the functions in Eq.~\eqref{eq:huang_science} with a constant expected error under certain criteria.
Recall that in Section~\ref{subsec:physical_systems} we argued that concepts based on physical systems can be used as a source of learning separations.
Since these concepts are of a similar form as the functions described in Eq.~\eqref{eq:huang_science}, one might wonder how the results of Huang et al.\ relate.

Let us take a closer look at the requirements of the methods described in~\cite{huang:science}.
Firstly, the Hamiltonians $H(x)$ must all be geometrically-local, and the observable $O$ must be a sum of polynomially many local observables $O = \sum_{i = 1}^L O_i$ such that $\sum_{i = 1}^L ||O_i||$ is bounded by a constant.
Additionally, the Hamiltonians $H(x)$ must all have a constant spectral gap (i.e., the difference between the smallest and the next smallest eigenvalue) and they must depend smoothly on $x$ (or more precisely, the average gradient of the function in Eq.~\eqref{eq:huang_science} must be bounded by a constant).
One might wonder what will happen if we relax the above requirements, while simultaneously maintaining the fact that a quantum computer would still be able to evaluate the function in Eq.~\eqref{eq:huang_science} (and hence build a learning separation around it based on Theorem~\ref{thm:seps_no-eff-data}).

Two possible relaxations of the requirements are the absence of a constant spectral gap (while maintaining an inverse polynomial spectral gap) and a reduced smoothness dependency of the Hamiltonian family on $x$ (i.e., compared to what is required for the methods of~\cite{huang:science}).
It turns out that if one relaxes these requirements, then under cryptographic assumptions the methods proposed by Huang et al.\ are no longer capable of evaluating the function in Eq.~\eqref{eq:huang_science} with constant expected error.
More precisely, any classical machine learning method that would still be able to evaluate the function in Eq.~\eqref{eq:huang_science} up to constant expected error under the relaxed assumptions would be able to solve $\mathsf{DLP}$ in $\mathsf{P/poly}$, which contradicts certain cryptographic assumptions. 
We provide a formal statement of this in the following theorem, the proof of which is deferred to Appendix~\ref{appendix:limitations_huang}.

\begin{restatable}{theorem}{limitationshuang}
\label{thm:limitations_huang}

Suppose there exists a polynomial-time randomized classical algorithm $\mathcal{A}$ with the following property: for every geometrically-local family of $n$-qubit Hamiltonians $H(x)$ there exist a dataset $\mathcal{T}_{H} \in \{0,1\}^{\mathrm{poly}(n)}$ such that for every sum $O = \sum_{i = 1}^{L}O_i$ of $L\in \mathcal{O}(\mathrm{poly}(n))$ many local observables with $\sum_{i = 1}^{L}||O_i||\leq B$ for some constant $B$, the function
\[
\overline{f}_{H, O}(x) = \mathcal{A}(x, O, \mathcal{T}_{H})
\]
satisfies
\[
    \mathbb{E}_{x \sim [-1, 1]^m}\Big[ \left|\overline{f}_{H, O}(x) - f_{H, O}(x) \right|\Big] < \frac{1}{6},
\]
where $f_{H, O}(x) = \mathrm{Tr}\left[\rho_H(x) O \right]$ and $\rho_H(x)$ denotes the ground state of $H(x)$. 
Then, $\mathsf{DLP} \in \mathsf{P/poly}$.
\end{restatable}

In conclusion, Theorem~\ref{thm:limitations_huang} shows that any method similar to that of~\cite{huang:science} cannot learn to predict ground state properties of certain physical systems discussed in Section~\ref{subsec:physical_systems}.
Moreover, there are a few subtle differences between the setup of~\cite{huang:science} and the one discussed in this paper.
Firstly, the classical shadow paradigm uses data that is different from the PAC learning setting (i.e., the data does not correspond to evaluations of the function it aims to predict).
This distinction in setup makes the approach of~\cite{huang:science} more versatile, as their data can be utilized to evaluate multiple different observables (moreover, their methods also work in the PAC setting).
Secondly, the functions $f_{H, O}$ in Eq.~\eqref{eq:huang_science} are real-valued, which differs from our paper where we investigate functions that map onto a discrete label space. 
It is possible to address this difference by applying a threshold function to $f_{H, O}$ after it is learned.
However, this thresholding introduces a mismatch in the types of data, as it would involve using real-valued data to learn a function with discrete values (which is clearly different from the PAC setting).

\subsection{Power of data}
\label{subsec:power_data}

In~\cite{huang:power} the authors show how having access to data radically enhances what can be efficiently evaluated.
In this section we connect the ideas from their work to the formalism we introduce in this paper.
Specifically, we will discuss a family of functions inspired by~\cite{huang:power} that from their description alone cannot be efficiently evaluated classically, yet access to a few examples (i.e., evaluations of the function) allows them to be efficiently evaluated classically.
This highlights an important difference between complexity-theoretic separations and learning separations, since in the latter one has to deal with the learner having access to data when proving classical non-learnability.

Consider a polynomial-depth parameterized quantum circuit $U(\theta, \vec{\phi})$ -- with two types of parameters $\theta \in \mathbb{R}$ parameterizing a single gate and $\vec{\phi} \in \mathbb{R}^\ell$ parameterizing multiple other gates -- that is universal in the sense that for every polynomial-depth circuit $V$ there exists parameters $\vec{\phi}^* \in \mathbb{R}^\ell$ such that 
\[
U(0, \vec{\phi}^*)\ket{0^n} = V\ket{0^n}.
\]
Moreover, assume the gates in $U$ are of the form $\mathrm{exp}\left(-\frac{i\theta}{2}A \right)$, with $A^2 = I$ (e.g., $Z$- or $X$-rotations).
By measuring the output of the circuit we define a family of single parameter functions given by
\[
f_{\vec{\phi}}(\theta) = \bra{0^n}U(\theta, \vec{\phi})^\dagger M U(\theta, \vec{\phi})\ket{0^n}.
\]

Following an argument similar to~\cite{huang:power}, due to the universality of the parameterized quantum circuit no efficient randomized classical algorithm can take as input a $\vec{\phi} \in \mathbb{R}^\ell$ and compute the function $f_{\vec{\phi}}$ on a given point $\theta \in \mathbb{R}$ up to constant error in time $\mathcal{O}\left(\mathrm{poly}(n)\right)$, unless $\mathsf{BPP} = \mathsf{BQP}$.
Intuitively, one might thus think that the concept class $\{f_{\vec{\phi}} \mid \vec{\phi} \in \mathbb{R}^\ell\}$ exhibits a separation between classical and quantum learners.
However, it turns out that the examples given to a classical learner radically enhance what it can efficiently evaluate.
In particular, given a few of evaluations of $f_{\vec{\phi}}$ for some fixed but arbitrary $\vec{\phi} \in \mathbb{R}^\ell$, a classical learner is suddenly able to efficiently evaluate the function.
To see this, note that by~\cite{nakanishi:pqc} one can write the functions as
\[
f_{\vec{\phi}}(\theta) = \alpha\cos(\theta - \beta) + \gamma, \quad \text{for }\alpha, \beta, \gamma \in \mathbb{R},
\]
where the coefficients $\alpha, \beta$ and $\gamma$ are all independent of $\theta$ (but they do depend on $\vec{\phi}$).
From this we can see that any three distinct examples $\big\{\big(\theta_i, f_{\vec{\phi}}(\theta_i)\big)\big\}_{i=1}^3$ uniquely determine $f_{\vec{\phi}}(\theta)$ and one can simply fit $\alpha, \beta$ and $\gamma$ to these three examples to learn how to evaluate $f_{\vec{\phi}}$ on unseen points.
We would like to point that the $\mathsf{BQP}$-hard problem in question is not evaluating $f_{\vec{\phi}}$ for a fixed $\vec{\phi} \in \mathbb{R}^\ell$, but rather evaluating $f_{\vec{\phi}}$ when $\vec{\phi}\in \mathbb{R}^\ell$ is part of the input.
This approach can be generalized to settings with more than one free parameter $\theta$, by using the fact that expectation values of parameterized quantum circuits can be written as a Fourier series~\cite{schuld:pqc}.
Specifically, when the number of frequencies appearing in the Fourier series is polynomial, then a polynomial number of examples suffices to fit the Fourier series and learn how to evaluate the expectation value of the quantum circuits for an arbitrary choice of parameters.

As discussed in Section~\ref{sec:eff_data}, one way to deal with the fact that data can radically enhance what can be efficiently evaluated is to ensure that the data itself is efficiently generatable.
However, for the concepts discussed above, the examples are such that only a quantum computer can generate them efficiently. 
In other words, these functions exemplify how hard to generate data can radically enhance what a classical learner can efficiently evaluate.
As discussed in Section~\ref{sec:eff_data}, another way to deal with the fact that data can radically enhance what can be efficiently evaluated is to ensure that the concepts lie outside of $\mathsf{HeurP}/\mathsf{poly}$.
However, for the case discussed above, every $f_{\vec{\phi}}$ corresponds to a function in $\mathsf{HeurP}/\mathsf{poly}$, since the coefficients $\alpha, \beta$ and $\gamma$ suffice as the advice.
Finally, we note that for certain circuits one could have exponentially many terms in the Fourier series~\cite{casas:fourier, caro:encoding}, in which case it is unclear how to classically learn them.

\subsection{Physically-motivated PAC learning settings with fixed hypothesis classes}
\label{subsec:proper_physics}

Throughout this paper we mainly focused on the setting where the learner is allowed to output arbibtrary hypotheses (barring that they have to be tractable as discussed in Appendix~\ref{appendix:poly_eval}). 
However, we want to highlight that setting where the learner is constrained to only be able to output hypothesis from a fixed hypothesis class is also relevant from a practical perspective.
In particular, in this section discuss how two well-studied problems (i.e., Hamiltonian learning, and identifying order parameters for phases of matter) fit in this setting.
Recall that in this setting, it is allowed and reasonable for the hypothesis class to be classically- or quantumly- intractable.

\paragraph{Hamiltonian learning}
In Hamiltonian learning one is given measurement data from a quantum experiment, and the goal is to recover the Hamiltonian that best matches the data.
Throughout the literature, various different types of measurement data have been considered.
For example, it could be measurement data from ground states, (non-zero temperature) thermal sates, or time-evolved states.
In our case, the data will be measurement data from time-evolved states and we formulate Hamiltonian learning in terms of a hypothesis class as follows.
First, we fix a (polynomially-sized) set of Hermitian operators $\{H_\ell\}_{\ell = 1}^L$.
Next, we consider a family of Hamiltonians $\{H_\beta\}_{\beta \in \mathbb{R}^L}$, where
\begin{align}
\label{eq:hams}
    H_\beta = \sum_{\ell = 1}^L \beta_\ell H_\ell.    
\end{align}
Finally, we define the hypothesis class $\mathcal{H}^{\mathrm{HL}} = \{h_\beta\}_{\beta \in \mathbb{R}^L}$, with concepts defined as
\begin{align}
\label{eq:hl_concepts}
    h_\beta(z, t) = \mathrm{sign}\big(\mathrm{Tr}\big[U^\dagger(t)\rho_z U(t)O_z\big]\big), \quad U(t) = e^{it H_\beta}.
\end{align}
Here $z$ describes the experimental setup, specifying the starting state (that will evolve under $H_\beta$ for time $t$) and the observable measured at the end.
A natural specification of the concepts that a learner could output are the parameters $\beta$.
In particular, in Hamiltonian learning we are only concerned with identifying which concept generated the data (i.e., what is the specification of the underlying Hamiltonian), as opposed to finding a hypothesis that closely matches the data.
In other words, the problem of Hamiltonian learning can naturally be formulated as PAC learning setting where the learner is constrained to only be able to output hypotheses described in Eq.~\eqref{eq:hl_concepts}.

With respect to learning separations, one might think that the above setting is a good candidate to exhibit a $\mathsf{C}_{\mathcal{H}^{\mathrm{HL}}}/\mathsf{Q}_{\mathcal{H}^{\mathrm{HL}}}$ separation, since the hypotheses are classically intractable and quantumly efficient to evaluate (assuming $\mathsf{BPP} \neq \mathsf{BQP}$).
Moreover, according to the folklore, quantum learners are most likely to have its advantages for data that is ``quantum-generated'', which certainly seems to be the case here.
However, recall that in the setting where the learner is constrained to output hypotheses from a fixed hypothesis class the task is not to evaluate, but rather to identify the concept generating the examples.
Therefore, the arguments we used throughout this paper do not directly apply.
In fact, it turns out that classical learners can efficiently identify the parameters of the Hamiltonian generating the data in many natural settings~\cite{anshu:hs, haah:hs, huang:hs}, eliminating the possibility of a $\mathsf{C}_{\mathcal{H}^{\mathrm{HL}}}/\mathsf{Q}_{\mathcal{H}^{\mathrm{HL}}}$ separation.
 
\paragraph{Order parameters and phases of matter}
When studying phases of matter one might want to identify what physical properties characterize the phase.
One can formulate this problem as finding a specification of the correct hypothesis selected from a hypothesis class consisting of possible \emph{order parameters}.
In particular, we fix the hypotheses $\mathcal{H}^{\mathrm{order}} = \{h_\alpha\}$ to be of a very special form, which compute certain expectation values of ground states given a specification of a Hamiltonian.
That is, we formally define the hypotheses as
\begin{align}
\label{eq:orderparam_concepts}
    h_\alpha(\beta) = \mathrm{sign}\left(\mathrm{Tr}\left[O_{\alpha}\rho_{\beta}\right]\right),
\end{align}
where $\rho_\beta$ denotes the ground state of some Hamiltonian specified by $\beta$ (e.g., using the parameterization in Eq.~\eqref{eq:hams}), and $\alpha$ specifies an observable $O_\alpha$ drawn from a set of observables that are deemed potential candidates for the order parameter that characterize the phase.
In this setting, one might not necessarily want to evaluate the hypotheses, as they might require one to prepare the ground state, which is generally intractable (even for a quantum computer).
However, one might still want to identify the observable $O_{\alpha}$ that correctly characterizes the phase of the physical system specified by $\beta$ (i.e., the corresponding \emph{order parameter}).
In other words, the problem of identifying order parameters naturally fits in the PAC learning setting where the learner is constrained to only be able to output hypotheses described in Eq.~\eqref{eq:orderparam_concepts}.

As in the case of Hamiltonian learning, one might think that the above concepts are good candidates to exhibit a $\mathsf{C}_{\mathcal{H}^{\mathrm{order}}}/\mathsf{Q}_{\mathcal{H}^{\mathrm{order}}}$ separation, since the hypotheses are classically intractable and quantumly efficient to evaluate (assuming $\mathsf{BPP} \neq \mathsf{BQP}$).
In fact, according to the folklore, quantum learners are most likely to have advantages for data that is ``quantum-generated'', which certainly seems to also be the case here.
However, as already mentioned, in the setting where the learner is constrained to output hypotheses from a fixed hypothesis class the goal is only to identify the correct hypothesis, and it is therefore not enough to just have concepts that are classically intractable.
We remark that the methods of~\cite{huang:science} also apply to phase classification, but they are more aimed at the PAC learning setting where the learner can output arbitrary hypotheses (i.e., the main goal is to predict the phase of a given physical system). 
In particular, their methods do not directly allow one to obtain a physically-meaningful description of the order-parameter, which is the main goal in the setting where the learner is constrained to output hypotheses from a fixed hypothesis class (which is related to the popular theme of ``explainability'' in machine learning).

\bigskip

In conclusion, while there has been progress in studying separations in the setting where the learner is constrained to output hypotheses from a fixed hypothesis class, there is still much to be discovered. 
Note that if the hypothesis class is $\mathsf{BQP}$-complete in the sense that it can perform arbitrary quantum computation, then a collapse similar to Lemma~\ref{lemma:cq=qq} happens and no separations are possible.
All in all, we have yet to find an example of a learning setting where the data is generated by a genuine quantum process and where it is necessary to use a quantum algorithm to efficiently identify the process generating the data.

\paragraph{Acknowledgements} 
The authors thank Simon C. Marshall, Srinivasan Arunachalam, and Tom O'Brien for helpful discussions.
The authors are grateful to Peter Bruin for valuable comments on the hardness of the discrete cube root for the $2^c$-integers in Section~\ref{subsec:new_sep}.
This work was supported by the Dutch Research Council (NWO/ OCW), as part of the Quantum Software Consortium programme (project number 024.003.037).
This work was also supported by the Dutch National Growth Fund (NGF), as part of the Quantum Delta NL programme.
This work was also supported by the European Union’s Horizon Europe program through the ERC CoG
BeMAIQuantum (Grant No.\ 101124342).
Views and opinions expressed are however those of the author(s) only and do not necessarily reflect those of the European Union or the European Research Council. 
Neither the European Union nor the granting authority can be held responsible for them.

\bibliographystyle{alpha}
\bibliography{main}

\clearpage
\appendix

\setcounter{footnote}{0}

\section{Details regarding definitions}
\label{appendix:details}

\subsection{Constraining hypothesis classes to those that are efficiently evaluatable}
\label{appendix:poly_eval}

In this section, we discuss why it makes sense to restrict the hypothesis class to be efficiently evaluatable.
Specifically, we show that if we allow the learner to use hypotheses that run for superpolynomial time, then every concept class that is learnable in superpolynomial time is also learnable in polynomial time.
Thus, if we do not restrict the hypotheses to be efficiently evaluatable, then the restriction that the learning algorithm has to run in polynomial time is vacuous (i.e., it imposes no extra restrictions on what can be learned).
For more details we refer to~\cite{kearns:clt}.

Consider a concept class $\mathcal{C}$ that is learnable by a superpolynomial time learning algorithm $\mathcal{A}$ using a hypothesis class $\mathcal{H}$.
To show that this concept class is learnable using a polynomial time learning algorithm, consider the hypothesis class $\mathcal{H}$ whose hypotheses are enumerated by all possible polynomially-sized sets of examples.
Each hypothesis in $\mathcal{H}'$ runs the learning algorithm $\mathcal{A}$ on its corresponding set of examples, and it evaluates the hypothesis from $\mathcal{H}$ that the learning algorithm outputs based on this set of examples.
Finally, consider the polynomial-time learning algorithm~$\mathcal{A}'$ that queries the example oracle a polynomial number of times and outputs the specification of the hypothesis in $\mathcal{H}'$ that corresponds to the obtained set of examples.
By construction, this polynomial-time learning algorithm $\mathcal{A}'$ now learns $\mathcal{C}$.

\subsection{Proof of Lemma~\ref{lemma:cq=qq}}
\label{appendix:cqqq}

\CQequalsQQ*

\begin{proof}

Since any efficient classical algorithm can be simulated using an efficient quantum algorithm it is obvious that $\mathsf{CQ} \subseteq \mathsf{QQ}$.
For the other inclusion, let $L = \left(\mathcal{C}, \mathcal{D}\right) \in \mathsf{QQ}$.
That is, the concept classes $\mathcal{C}$ are efficiently learnable under the distributions $\mathcal{D}$ by a quantum learning algorithm $\mathcal{A}^q$ using a quantum evaluatable hypothesis class $\mathcal{H}$. 
To show that $L \in \mathsf{CQ}$, consider the quantum evaluatable hypothesis class $\mathcal{H}'$ whose hypotheses are enumerated by all possible polynomially-sized sets of training examples.
Each hypothesis in $\mathcal{H}'$ runs the quantum learning algorithm $\mathcal{A}^q$ on its corresponding set of examples, and evaluates the hypothesis from $\mathcal{H}$ that the quantum learning algorithm $\mathcal{A}^q$ outputs based on the set of examples.
Finally, consider the classical polynomial-time learning algorithm $\mathcal{A}^c$ that queries the example oracle a polynomial number of times and outputs the specification of the hypothesis in $\mathcal{H}'$ that corresponds to the obtained set of examples.
By construction, this classical polynomial-time algorithm $\mathcal{A}^c$ can learn the concept classes $\mathcal{C}$ under the distributions $\mathcal{D}$ using the quantum evaluatable hypothesis class $\mathcal{H}'$. 
This shows that $L \in \mathsf{CQ}$.

\end{proof}

\subsection{Proof of Lemma~\ref{lemma:samppoly}}
\label{appendix:samp_poly}

\samppoly*

\begin{proof}
The proof strategy is analogous to the arguments in Section 2 of the supplementary material of~\cite{huang:power}, where they show that $\mathsf{BPP/samp} \subseteq \mathsf{P/poly}$.
Let $(L, \{\mathcal{D}_n\}_{n \in \mathbb{N}})~\in~\mathsf{HeurBPP/samp}$.
In particular, there exist a polynomial-time classical algorithms $\mathcal{A}$ with the following property: for every $n$ and $\epsilon > 0$ it holds that
\begin{align}
\label{eq:heursamp1}
    \mathsf{Pr}_{x \sim \mathcal{D}_n}\left[\mathsf{Pr}\big(\mathcal{A}(x, 0^{\lfloor 1/\epsilon \rfloor}, \mathcal{T}) = L(x)\big) \geq \frac{2}{3}\right] \geq 1 - \epsilon,
\end{align} 
where the inner probability is over the randomization of $\mathcal{A}$ and $\mathcal{T} = \{(x_i, L(x_i)) \mid x_i \sim \mathcal{D}_n\}_{i=1}^{\mathrm{poly}(n)}$.

Let $\epsilon > 0$ and partition the set of $n$-bit strings as follows
\begin{align}
    \label{eq:split}
\{0,1\}^n = I^n_{\text{correct}}(\epsilon) \sqcup I^n_{\text{error}}(\epsilon),
\end{align}
such that for every $x \in I^n_{\text{correct}}(\epsilon)$ we have
\begin{align}
\mathsf{Pr}\big(\mathcal{A}(x, 0^{\lfloor 1/\epsilon \rfloor}, \mathcal{T}) = L(x)\big) \geq \frac{2}{3},
\end{align}
where the probability is taken over the internal randomization of $\mathcal{A}$ and $\mathcal{T}$.
Importantly, we remark that our partition is such that
\begin{align}
\label{eq:probeps}
\mathsf{Pr}_{x \sim \mathcal{D}_n}\big[x \in I^n_{\text{correct}}(\epsilon)\big] \geq 1-\epsilon.
\end{align}

By applying the arguments of Section 2 of the supplementary material of~\cite{huang:power} to $\mathcal{A}$ with the bitstring $0^{\lfloor 1/\epsilon \rfloor}$ fixed as input we obtain  a polynomial-time classical algorithm $\mathcal{A}'$ with the following property: for every $n$ there exists an advice string $\alpha_{n, \epsilon} \in \{0,1\}^{\mathrm{poly(n, 1/\epsilon)}}$\footnote{Note that the advice string also depends on $\epsilon$, since $0^{\lfloor 1/\epsilon \rfloor}$ was fixed as input to $\mathcal{A}$.} such that for every $x \in  I^n_{\text{correct}}(\epsilon)$:
\begin{align}
\mathcal{A}'(x, 0^{\lfloor 1/\epsilon \rfloor}, \alpha_{n, \epsilon}) = L(x).
\end{align}
Intuitively, the algorithm $\mathcal{A}'(x, 0^{\lfloor 1/\epsilon \rfloor}, \alpha_{n, \epsilon})$ runs $\mathcal{A}(x, 0^{\lfloor 1/\epsilon \rfloor}, \mathcal{T})$ a certain number of times and decides its output based on a majority-vote.
Moreover, $\mathcal{A}'$ does so with a particular setting of random seeds and training data $\mathcal{T}$ that makes it correct decide $L(x)$, which is collected in $\alpha_{n, \epsilon}$.
Finally, from Eq.~\eqref{eq:probeps} we find that we have
\begin{align}
\mathsf{Pr}_{x \sim \mathcal{D}_n}\left[\mathcal{A}'(x, 0^{\lfloor 1/\epsilon \rfloor}, \alpha_{n, \epsilon}) = L(x)\right] \geq 1 - \epsilon,
\end{align}
which shows that $(L, \{\mathcal{D}_n\}_{n \in \mathbb{N}}) \in \mathsf{HeurP/poly}$.

\end{proof}

\section{Discrete logarithm concept class}
\subsection{$\mathsf{CC}/\mathsf{QQ}$ separation from intractability assumption in~\cite{blum:hardness}}
\label{appendix:dlp1}

In this section we show how to construct a binary concept class that achieves a $\mathsf{CC/QQ}$ separation when one assumes the classical intractability of the discrete logarithm class as stated in~\cite{blum:hardness}.
More precisely, in~\cite{blum:hardness} they state that if a family of classical circuits $\{C_n\}$ satisfies
\begin{align}
\label{eq:solvep}
C_n(x, g, p) = \log_{(a,p)} (x)\quad \forall a \in \{1, \dots, p -1\}\text{ and }\forall x \in \mathbb{Z}^*_p.
\end{align}
for at least a $\frac{1}{\mathrm{poly}(n)}$ fraction of all primes $p$, then the size of $C_n$ grows faster than any polynomial~in~$n$.

\begin{remark}
$\log_{(a,p)} x$ denotes the discrete logarithm modulo $p$ of $x$ with respect to the generator~$a$. That is, the discrete logarithm $\log_{(a, p)} x$ is the smallest positive integer $\ell$ such that $a^\ell \equiv x \mod p.$
\end{remark}

\dlponedef*

\dlpone*

\begin{proof}

For quantum learnability, we note that a quantum learner can with high probability recover the tuple $(a, p) \in \{0,1\}^{2n}$ from $\mathcal{O}\left(\mathrm{poly}(n)\right)$ examples drawn from $E(c_{(a,p)}, \mathcal{D}_n^{U})$ (i.e., it just considers those for which $b = 1$).
Once the quantum learner has figured out what $(a, p)$ is, it can directly use Shor's algorithm~\cite{shor:factoring} to compute the concept $c_{(a, p)}$ on unseen datapoints $(b, j, x) \in \mathcal{X}_n$.

For classical non-learnability, we note that since the examples are efficiently generatable classically, the existence of a classical learner implies the existence of poly($n$)-sized classical circuits $C_n$ that satisfy
\begin{align}
\label{eq:performance}
    \mathsf{Pr}_{(b, j, x)\sim \mathcal{D}_n^U}\left[C_n(a, p, b, j, x) = c_{(a, p)}(b, j, x)\right] \geq 1 - \epsilon, \quad\text{for some }\epsilon \in \Omega(1/\mathrm{poly}(n)).
\end{align}
In particular, the circuit $C_n$ generates examples for $c_{(a, p)}$ based on its input $(a, p)$, runs the classical learning algorithm and afterwards evaluates the hypothesis output by the learner on $(b, j, x)$.
What remains to be shown is that this circuit $C_n$ can be used to compute $\mathds{1}_{[0, (p-1)/2]}(\log_{(a, p)}(x))$ on at least a $\frac{1}{2} + \frac{1}{\mathrm{poly}(n)}$ fraction of $x \in \mathbb{Z}_p^*$, since by the random self-reducibility outlined in~\cite{blum:hardness} this would violate the hardness of the discrete logarithm stated in~\cite{blum:hardness}.
To see why this holds, note that the worst $C_n$ can do for computing $\mathds{1}_{[0, (p-1)/2]}(\log_{(a, p)}(x))$ while satisfying Eq.~\eqref{eq:performance} is when it is correct on all inputs with $b=1$, and if it is correct on at least one $j'$ for a given $x'$ then it must be correct on all $j = 1, \dots, n$ for that $x'$.
However, note that even in this case $C_n$ can correctly compute $\mathds{1}_{[0, (p-1)/2]}(\log_{(a, p)}(x))$ on at least a $1 - 2\epsilon$ fraction of all possible $x$.
Finally, as $\epsilon \in \Omega(1/\mathrm{poly}(n))$, we conclude that $C_n$ can compute $\mathds{1}_{[0, (p-1)/2]}(\log_{(a, p)}(x))$ on at least a $\frac{1}{2} + \frac{1}{\mathrm{poly}(n)}$ fraction of $x \in \mathbb{Z}_p^*$, thereby violating the hardness of the discrete logarithm stated in~\cite{blum:hardness}.

\end{proof}

\subsection{Singleton $\mathsf{CC}/\mathsf{QQ}$ separation from intractability assumption in~\cite{blum:hardness}}
\label{appendix:dlp2}

In this section we show how to construct a non-binary singleton concept class that achieves a $\mathsf{CC/QQ}$ separation when one assumes the classical intractability of the discrete logarithm class as stated in~\cite{blum:hardness} (see previous section for the formal statement of intractability stated in~\cite{blum:hardness}).

\dlptwodef*

\smallskip

\dlptwo*

\begin{proof}
    For quantum learnability, we note that the quantum learner does not need data, as it can output the hypothesis that reads inputs $(a, p, x)$ and uses Shor's algorithm to compute $\log_{(a, p)}(x)$.

    \enskip

    \noindent For classical non-learnability, we first recap the following claims about ``random self-reducibility''.

    \smallskip 

\noindent \textbf{Claim 1}: Fix a prime $p$, due to ``random self-reducibilities'' we have
\begin{itemize}
    \item For any fixed $a$, any polynomial-time algorithm $\mathcal{A}$ that satisfies
    \[
    \mathsf{Pr}_{x \in_U \mathbb{Z}_p^*} \left(\mathcal{A}(a, p, x) = \log_{(a, p)}(x) \right) \geq \frac{1}{2} + \frac{1}{\mathrm{poly}(n)}
    \]
    can be turned into a randomized  polynomial-time algorithm $\mathcal{A}'$ that satisfies
    \[
        \mathcal{A}'(a, p, x) = \log_{(a, p)}(x), \quad \forall x \in \mathbb{Z}^*_p.
    \]
    \noindent \underline{\emph{proof of claim}}: $\mathcal{A}'$ on input $(a, p, x)$ samples a uniformly random $c \in \{0, \dots, p-1\}$ and computes $ b = \mathcal{A}(a, p, xa^c)$. 
    If $a^{b - c} = x$, output $b-c$. 
    Otherwise, resample $c \in \{0, \dots, p-1\}$ and repeat.

    \item Any polynomial-time algorithm $\mathcal{A}$ that satisfies
    \[
    \mathsf{Pr}_{a \in_U \mathbb{Z}_p^*} \Big( \forall x \in \mathbb{Z}_p^* \text{ }:\text{ }\mathcal{A}(a, p, x) = \log_{(a, p)}(x)\Big) \geq \frac{1}{2} + \frac{1}{\mathrm{poly}(n)}
    \]
    can be turned into an polynomial-time algorithm $\mathcal{A}'$ that satisfies
    \[
        \mathcal{A}'(a, p, x) = \log_{(a, p)}(x) , \quad \forall a \in \mathbb{Z}^*_p.
    \]

    \noindent \underline{\emph{proof of claim}}: $\mathcal{A}'$ on input $(a, p, x)$ samples a uniformly random $a' \in \{0, \dots, p-1\}$ and computes $b = \mathcal{A}(a', p, a)$ and $c = \mathcal{A}(a', p, x)$. If $\mathcal{A}$ is correct on $a'$, i.e., we have that
    \[
    (a')^b = a,\text{ }(a')^c = x\text{ and }\gcd(b, p-1) = 1,
    \]
    then output $b^{-1}c$. 
    Otherwise, resample $a' \in \{0, \dots, p-1\}$ and repeat.
    
\end{itemize}

\medskip

\noindent \textbf{Claim 2}: Fix a prime $p$, a result of Claim 1 is that any polynomial-time algorithm $\mathcal{A}$ that satisfies
\begin{align}
    \label{eq:performance_fixed}
    \mathsf{Pr}_{(a, x) \in_U \left(\mathbb{Z}_p^*\right)^2} \Big(\mathcal{A}(a, p, x) = \log_{(a,p)} (x)) \Big) \geq \frac{3}{4} + \frac{1}{P(n)}
\end{align}
for a polynomial $P$ can be turned into a polynomial-time algorithm $\mathcal{A}'$ that satisfies
\begin{align}
\label{eq:solvep2}
\mathcal{A}'(a, p, x) = \log_{(a,p)} (x)\quad \forall a \in \{1, \dots, p -1\}\text{ and }\forall x \in \mathbb{Z}^*_p.
\end{align}

\smallskip

\noindent \underline{\emph{proof of claim}}: Let $P$ be a polynomial, and suppose $\mathcal{A}$ satisfies Eq.~\eqref{eq:performance_fixed}.
We say ``$\mathcal{A}$ solves $a$'' if
\[
    \label{eq:solvesg}
    \mathsf{Pr}_{x \in_U \mathbb{Z}_p^*} \left(\mathcal{A}(a, p, x) = \log_{(a,p)} (x) \right) > \frac{1}{2} + \frac{1}{P(n) - 2}.
\]
That is, the number of tuples $(a, x)$ such that $\mathcal{A}(a, p, x)$ is incorrect is at most $(\frac{1}{2} - \frac{1}{P(n)-2})\cdot 2^n$

The question we now ask ourselves is: what is the smallest amount of $a$ that $\mathcal{A}$ can solve while satisfying Eq.~\eqref{eq:performance_fixed}?
The number of $a$ that $\mathcal{A}$ cannot solve is maximized if $\mathcal{A}$ barely fails to solve it on all the $a$ that it cannot solve.
Specifically, for every $a$ that $\mathcal{A}$ cannot solve we have
\[
    \label{eq:frugalsolve}
    \mathsf{Pr}_{x \in_U \mathbb{Z}_p^*} \left(\mathcal{A}(a, p, x) = \log_{(a,p)} (x) \right) = \frac{1}{2} + \frac{1}{P(n)-2}.
\]
That is, for every $a$ that $\mathcal{A}$ cannot solve the number of tuples $(a, x)$ such that $\mathcal{A}(a, p, x)$ is incorrect is $(\frac{1}{2} - \frac{1}{P(n)}-2)2^n$.
Let $f$ denote the number of $a$ that $\mathcal{A}$ cannot solve, from Eq.~\eqref{eq:performance_fixed} we get
\[
f \cdot \left(\frac{1}{2} - \frac{1}{P(n)-2}\right)2^n < \left(\frac{1}{4} - \frac{1}{P(n)}\right)2^{2n}.
\]
which shows that $f < \left(\frac{1}{2} - \frac{1}{P(n)}\right)2^n$.

Part 1 of Claim 1 implies that if $\mathcal{A}$ solves $a$, then $\mathcal{A}$ can be turned into a polynomial-time randomized algorithm $\mathcal{A}'$ that is correct on all $x$ for the given generator $a$.
Finally, since the number of $a$ that $\mathcal{A}$ cannot solve is less than $\left(\frac{1}{2} - \frac{1}{P(n)}\right)2^n$, part 2 of Claim 1 implies that $\mathcal{A}'$ can be turned into a polynomial-time randomized algorithm $\mathcal{A}''$ that satisfies
\begin{align}
\label{eq:solvep3}
\mathcal{A}''(a, p, x) = \log_{(a,p)} (x)\quad \forall a \in \{1, \dots, p -1\}\text{ and }\forall x \in \mathbb{Z}^*_p.
\end{align}

\bigskip 

Using the above two claims, we are able to show that no classical learning algorithm can achieve ``efficient learnability''.
In particular, the existence of an efficient classical learner implies the existence of a polynomial--time algorithm $\mathcal{A}$ that satisfies
\begin{align}
\label{eq:performance_total}
    \mathsf{Pr}_{(a, p, x)}\left(\mathcal{A}(a, p, x) =  \log_{(a,p)} (x)\right) \geq 1 - \frac{1}{\mathrm{poly}(n)}
\end{align}
which we can show -- using Claims 1 and 2 -- can be turned into an algorithm $\mathcal{A}'$ that satisfies
\begin{align}
\label{eq:solvep4}
\mathcal{A}'(a, p, x) = \log_{(a,p)} (x)\quad \forall a \in \{1, \dots, p -1\}\text{ and }\forall x \in \mathbb{Z}^*_p
\end{align} 
for more than a polynomial fraction of all primes (violating the intractability as outlined in~\cite{blum:hardness}).

To see this, we suppose that an algorithm $\mathcal{A}$ satisfies Eq.~\eqref{eq:performance_total}, and show that it will satisfy Eq.~\eqref{eq:performance_fixed} for at least a $1 - 1/\mathrm{poly}(n)$ fraction of all primes, which violates the intractability of the discrete logarithm as outlined in~\cite{blum:hardness}.
We do so by letting $E$ denote the set of triples $(a, p, x)$ on which $\mathcal{A}$ is incorrect and asking ourselves what the most ``prime costly'' way for $\mathcal{A}$ to fail (note here we say that $\mathcal{A}$ fails on some prime if it does not satisfy Eq.~\eqref{eq:performance_fixed} for this prime).
Here by ``prime costly'' we mean: \emph{what is the cheapest way (in terms of number of triples in $E$) for $\mathcal{A}$ to fails on a given prime $p'$.}
In other words, what is the largest amount of primes on which $\mathcal{A}$ can fail.

In short, we note that the most ``prime costly'' would be to barely fail on every prime $p$ on which it fails. 
That is, for every prime $p'$ on which $\mathcal{A}$ fails, we have that
\[
\# \{(a, p', x) \in E\} = \left(\frac{3}{4} + \frac{1}{\mathrm{poly(n)}}\right) \cdot 2^{2n}
\]
Finally, this leads us to the observation that
\[
\#\text{primes $p$ on which $\mathcal{A}$ fails} \leq \frac{|E|}{\left(\frac{1}{4} + \frac{1}{\mathrm{poly(n)}} \right) \cdot 2^{2n}} \leq 2^n \cdot \frac{1}{\mathrm{poly(n)}}
\]
where we use that $|E| \leq \left(1 - \frac{1}{\mathrm{poly(n)}}\right)\cdot 2^{3n}$.

\end{proof}

\section{Discrete cube root concept class}
\subsection{$\mathsf{CC}/\mathsf{QC}$ separation from intractability assumption in~\cite{kearns:clt}}
\label{appendix:dcr1}

In this section we show how to construct a binary concept class that achieves a $\mathsf{CC/QC}$ separation when one assumes the classical intractability of the discrete cube root as stated in~\cite{kearns:clt}.
More precisely, in~\cite{kearns:clt} they state that if a family of classical circuits $\{C_n\}$ satisfies
\begin{align}
\label{eq:solvep}
C_n(N, x) = f_N^{-1}(x)\quad \forall x \in \mathbb{Z}^*_N.
\end{align}
for at least a $\frac{1}{\mathrm{poly}(n)}$ fraction of moduli $N$, then the size of $C_n$ grows faster than any polynomial~in~$n$.

\smallskip

\dcronedef*

\smallskip

\dcrone*

\begin{proof}

For quantum learnability, we note that a quantum learner can with high probability recover the modulus $N$ from $\mathcal{O}\left(\mathrm{poly}(n)\right)$ examples drawn from $E(c_{(N,i)}, \mathcal{D}_n^{U})$ (i.e., it just considers those for which $b = 1$).
Once the quantum learner has figured out what $N$ is, it can directly use Shor's algorithm~\cite{shor:factoring} to compute the $d^*$ such that $f_N^{-1}(x) \equiv x^{d^*} \mod N$.

For classical non-learnability, we note that since the examples are efficiently generatable classically, the existence of a classical learner implies the existence of poly($n$)-sized classical circuits $C_n$ that satisfy
\begin{align}
\label{eq:performancedcr}
    \mathsf{Pr}_{(b, j, x)\sim \mathcal{D}_n^U}\left[C_n(N, i, b, j, x) = c_{(N, i)}(b, j, x)\right] \geq 1 - \epsilon, \quad\text{for some }\epsilon \in \Omega(1/\mathrm{poly}(n)).
\end{align}
In particular, the circuit $C_n$ generates examples for $c_{(N, i)}$ based on its input $(N, i)$, runs the classical learning algorithm and afterwards evaluates the hypothesis output by the learner on $(b, j, x)$.
What remains to be shown is that this circuit $C_n$ can be used to compute $\mathrm{bin}(f_N^{-1}(x), n)$ on at least a $\frac{1}{2} + \frac{1}{\mathrm{poly}(n)}$ fraction of $x \in \mathbb{Z}_p^*$, since by the random self-reducibility outlined in~\cite{alexi:rsa, goldwasser:rsa} this would violate the hardness of the discrete logarithm stated in~\cite{kearns:clt}.
To see why this holds, note that the worst $C_n$ can do for computing $\mathrm{bin}(f_N^{-1}(x), n)$ while satisfying Eq.~\eqref{eq:performancedcr} is when it is correct on all inputs with $b=1$, and if it is correct on at least one $j'$ for a given $x'$ then it must be correct on all $j = 1, \dots, n$ for that $x'$.
However, note that even in this case $C_n$ can correctly compute $\mathrm{bin}(f_N^{-1}(x), n)$ on at least a $1 - 2\epsilon$ fraction of all possible $x$.
Finally, as $\epsilon \in \Omega(1/\mathrm{poly}(n))$, we conclude that $C_n$ can compute $\mathrm{bin}(f_N^{-1}(x), n)$ on at least a $\frac{1}{2} + \frac{1}{\mathrm{poly}(n)}$ fraction of $x \in \mathbb{Z}_N^*$, thereby violating the hardness of the discrete logarithm stated in~\cite{kearns:clt}.

\end{proof}

\subsection{Singleton $\mathsf{CC}/\mathsf{QQ}$ separation from intractability assumption in~\cite{kearns:clt}}
\label{appendix:dcr2}

In this section we show how to construct a binary singleton concept class that achieves a $\mathsf{CC/QQ}$ separation when one assumes the classical intractability of computing the function $f_N^{-1}$ as outlined in~\cite{kearns:clt} (see the previous section for the formal statement of intractability stated in~\cite{kearns:clt}).

\dcrtwodef*

\smallskip

\dcrtwo*

\begin{proof}
    For quantum learnability, we note that the quantum learner does not need data, as it can output the hypothesis that reads inputs $(N, x)$ and uses Shor's algorithm to compute $f_N^{-1}(x)$.

    For classical non-learnability, we note that since the examples are efficiently generatable classically, the existence of an efficient classical learner implies the existance of a polynomial-time classical algorithm $\mathcal{A}$ that satisfies
    \begin{align}
        \mathsf{Pr}_{x \in \mathbb{Z}_N^*}\left[\mathcal{A}(N, x) =  \mathrm{bin}(f^{-1}_{N}(x), n)\right] \geq \frac{1}{2} + \frac{1}{\mathrm{poly}(n)}. 
    \end{align}
    for at least an inverse polynomial fraction of moduli $N$.
    Since computing the least-significant bit of $f_N^{-1}$ on a $\frac{1}{2} + \frac{1}{\mathrm{poly}(n)}$ fraction of $x \in \mathbb{Z}_N^*$ is as hard as to computing $f^{-1}_N$ on all $x\in \mathbb{Z}_N^*$~\cite{alexi:rsa, goldwasser:rsa}, we conclude that $\mathcal{A}$ can be turned into a polynomial-time algorithm that computes $f^{-1}_{N}$ for an inverse polynomial fraction of moduli, violating the hardness assumption outlined in~\cite{kearns:clt}.  
\end{proof}

\section{Proof of Theorem~\ref{thm:modexp}}
\label{appendix:modexp}

\modexp*

\begin{proof}

To see why $L_{\mathrm{modexp}}$ is not in $\mathsf{CC}$, we first note that the modular exponentiation concept class contains the cube root function $f_N^{-1}$ discussed in Section~\ref{subsec:3root}.
Therefore, the proof presented in~\cite{kearns:clt}, which shows that the cube root concept class is not in $\mathsf{CC}$ under the DCRA, also implies that the modular exponentiation concept class is not in $\mathsf{CC}$ under the DCRA.
To briefly recap, recall from Section~\ref{subsec:3root} that we can efficiently generate examples $(y, f_N^{-1}(y))$, for $y \in \mathbb{Z}_N^*$ uniformly at random.
If we put these examples (together with examples for $b=1$) into an efficient classical learning algorithm for the modular exponentiation concept class, the learning algorithm would with high probability identify a classically efficiently evaluatable hypothesis that agrees with $f_N^{-1}$ on a polynomial fraction of inputs.
This directly violates the $2^c$-DCRA, which states that evaluating $f_N^{-1}$ is classically intractable, even on an inverse polynomial fraction of inputs.

What remains to be shown is that $L_{\mathrm{modexp}}$ is in $\mathsf{QC}$.
First, we note that a quantum learner can with high probability recover the modulus $N$ from $\mathcal{O}\left(1\right)$ examples drawn from $E(c_{(N,d)}, \mathcal{D}_n^{U})$ (i.e., it just needs one for which $b = 1$).
Next, we note that $EX(c_d, \mathcal{D}^U_n)$ will in $\mathcal{O}(1)$ queries with high probability return an example $(x, x^d \mod N)$, where $x$ is sampled uniformly at random from $\mathbb{Z}_N^*$.
To show that $L_{\mathrm{modexp}}$ is in $\mathsf{QC}$, we will describe a $\mathcal{O}(\mathrm{poly}(n, 1/\delta))$-time quantum learning algorithm that uses queries to $EX(c_d, \mathcal{D}^U_n)$ and identifies $d$ with probability at least $1-\delta$.
Before describing the quantum learning algorithm, we will first prove the following lemma, which we will later use to help us prove that our learning algorithm is both correct and efficient.

\begin{lemma}
\label{lemma:order}
Write $(p-1)(q-1) = 2^c \cdot p^{k_1}_1 \cdots p^{k_\ell}_\ell$, where the $p_i$s are distinct odd primes.
Then, for any $i = 1, \dots, \ell$ we have that with probability at least $1/2$, an example $(x, x^d \mod N)$ queried from $EX(c_d, \mathcal{D}^U_n)$ satisfies
\begin{align}
\label{eq:order}
    p^{k_i}_i \mid \mathrm{ord}_N(x),
\end{align}
where $\mathrm{ord}_N(x)$ denotes the order of $x$ in $\mathbb{Z}_N^*$.
\end{lemma}
\begin{proof}[Proof of Lemma~\ref{lemma:order}]

Let $C_k$ denote the cyclic group of order $k$.
Since $(p-1)(q-1) = 2^ca$ and $\gcd(p-1, q-1) = 2^{c'}$ as described in Definition~\ref{def:modexp}, the Chinese remainder theorem tells us that
\[
\mathbb{Z}_N^* \simeq \mathbb{Z}_p^* \times \mathbb{Z}_q^* \simeq C_{p-1} \times C_{q-1} \simeq C_{2^{c_1}} \times C_{2^{c_2}} \times C_{p^{k_1}_1} \times \dots \times C_{p^{k_\ell}_\ell},
\]
for some $c_1, c_2$ such that $c_1 + c_2 = c$.
For any $x \in \mathbb{Z}_N^*$ we will use $(x^{(1)}_0, x^{(2)}_0, x_1, \dots, x_\ell)$ to denote its corresponding element in $C_{2^{c_1}} \times C_{2^{c_2}} \times C_{p^{k_1}_1} \times \dots \times C_{p^{k_\ell}_\ell}$.
Next, note that the order of $x$ in $\mathbb{Z}_N^*$ is the least common multiple of the orders of all $x^{(1)}_0, x^{(2)}_0, x_1, \dots, x_\ell$ in their respective groups.
What this implies is that any element $x = (x^{(1)}_0, x^{(2)}_0, x_1, \dots, x_\ell)$ satisfies 
\[
    p^{k_i}_i \mid \mathrm{ord}_N(x),
\]
if $x_i$ is a generator of $C_{p^{k_i}_i}$.
The number of generators of $C_{p^{k_i}_i}$ is equal to $\varphi(p^{k_i}_i)$ (where $\varphi$ denotes Euler's totient function), and the number of elements of $x = (x^{(1)}_0, x^{(2)}_0, x_1, \dots, x_\ell)$ such that $x_i$ is a generator of $C_{p^{k_i}_i}$ is therefore equal to 
\[
2^{c_1} \cdot 2^{c_2} \cdot p^{k_1}_1 \cdots p^{k_{i-1}}_{i-1} \cdot \varphi(p^{k_i}_i) \cdot p^{k_{i+1}}_{i+1} \cdots p^{k_\ell}_\ell.
\]
Thus, the probability that a uniformly random $x \in \mathbb{Z}_N^*$ satisfies Eq.~\eqref{eq:order}
 is at least
 \begin{align*}
    \frac{2^{c_1} \cdot 2^{c_2} \cdot p^{k_1}_1 \cdots p^{k_{i-1}}_{i-1} \cdot \varphi(p^{k_i}_i) \cdot p^{k_{i+1}}_{i+1} \cdots p^{k_\ell}_\ell.}{\# \mathbb{Z}_N^*} &= \frac{2^{c_1} \cdot 2^{c_2} \cdot p^{k_1}_1 \cdots p^{k_{i-1}}_{i-1} \cdot \varphi(p^{k_i}_i) \cdot p^{k_{i+1}}_{i+1} \cdots p^{k_\ell}_\ell.}{(p-1)(q-1)} \\
    &= \frac{\varphi(p^{k_i}_i)}{p^{k_i}_i} \geq \frac{1}{2}.
 \end{align*}
\end{proof}

We now describe the quantum learning algorithm that can identify $d$ in time $\mathcal{O}(\mathrm{poly}(n, 1/\delta))$ using queries to $EX(c_d, \mathcal{D}_n^U)$.
We write $(p-1)(q-1) = 2^c \cdot p^{k_1}_1 \cdots p^{k_\ell}_\ell$, where the $p_i$s are distinct primes.
The idea is to query $EX(c_d, \mathcal{D}_n^U)$ sufficiently many times such that for every $i = 1, \dots, \ell$ we have an example $(x_i, x_i^d \mod N)$ where
\begin{align}
\label{eq:cong}
    p^{k_i}_i \mid \mathrm{ord}_N(x_i).
\end{align}
Next, we use Shor's algorithm~\cite{shor:factoring} to compute $r_i= \mathrm{ord}_N(x_i)$ and $a_i = \log_{x_i}(x_i^d)$, where $\log_{a}(b)$ denotes the discrete logarithm of $b$ in the group generated by $a$ (i.e., the smallest integer $\ell$ such that $a^\ell = b$).
Now by elementary group theory we obtain the congruence relation
\[
    d \equiv a_i \mod r_i,
\]
which by Eq.~\eqref{eq:cong} implies the congruence relation
\[
    d \equiv a_i \mod p^{k_i}_i.
\]
In other words, the examples allowed us to recover $d \mod p^{k_i}_i$ for every $i = 1, \dots, \ell$.
By the Chinese remainder theorem, all that remains is to recover $d \mod 2^c$, which we can do by brute force search since $c$ is constant. 
All in all, if we query $EX(c_d, \mathcal{D}_n^U)$ sufficiently many times such that for every $i = 1, \dots, \ell$ we have an example $(x_i, x_i^d \mod N)$ satisfying Eq.~\eqref{eq:cong}, then we can recover $d$.

What remains to be shown is that with probability $1 - \delta$ a total number of $\mathcal{O}(\mathrm{poly}(n, 1/\delta))$ queries to $EX(c_d, \mathcal{D}_n^U)$ suffices to find an example $(x_i, x_i^d \mod N)$ satisfying Eq.~\eqref{eq:cong} for every $i = 1, \dots, \ell$.
To do so, we invoke Lemma~\ref{lemma:order} and conclude from it that for any individual $i = 1, \dots, \ell$ after $\mathcal{O}(\log(1/\delta'))$ queries with probability at least $1-\delta'$ we found an example $(x_i, x_i^d \mod N)$ satisfying Eq.~\eqref{eq:cong}.
In particular, this implies that after a total of $\mathcal{O}(\log(n, 1/\delta))$ queries we found with probability at least $1- \delta$ examples $(x_i, x_i^d \mod N)$ satisfying Eq.~\eqref{eq:cong} for all $i = 1, \dots, \ell$.

\end{proof}

\subsection{Discrete cube root assumption for moduli of Definition~\ref{def:modexp}}
\label{appendix:moduli}

Recall that in Definition~\ref{def:modexp} we have constraint our moduli $N=pq$ to satisfy the conditions
\begin{itemize}
    \item[$(a)$] $\gcd(3, (p-1)(q-1)) = 1$,
    \item[$(b)$] $(p-1)(q-1) = 2^c\cdot a$, where $a \in \mathbb{N}$ is odd and $c$ is a constant,
    \item[$(c)$] $\gcd(p-1, q-1) = 2^{c'}$ for some $c'$.
\end{itemize}
Firstly, we remark that $(a)$ is a standard condition required for the function cube root function $f_N^{-1}$ to be well-defined, and it therefore does not influence the DCRA.
On the other hand, the implications that the conditions $(b)$ and $(c)$ have on the hardness in the DCRA are relatively unexplored.
Nonetheless, there are reasons to believe that the DCRA still holds under conditions $(b)$ and $(c)$.

To see why conditions $(b)$ and $(c)$ might not influence the DCRA, we remark that the DCRA is closely-related to the security of the RSA cryptosystem.
Specifically, the DCRA for a specific modulus $N$ is equivalent to assuming that the RSA cryptosystem with public exponentiation key $e = 3$ and modulus $N$ has an ``exponential security'' (i.e., deciphering a ciphertext without the private key requires time exponential in the cost of decryption).
In other words, if a certain family of moduli is used in practice, or are not actively avoided, this can be considered as supporting evidence that the DCRA holds for those moduli.

In practice it is generally prefered to use so-called ``cryptographically strong primes''\footnote{\url{https://en.wikipedia.org/wiki/Strong_prime}} $p$ and $q$ when constructing the modulus $N = pq$ for the RSA cryptosystem.
One of the conditions for a prime $p$ to be a cryptographically strong prime is that $p-1$ has large prime factors.
Note that if $p-1$ has large prime factors, then the largest power of 2 that divides it must be small. 
In other words, if $p$ and $q$ are cryptographically strong primes, then condition $(b)$ holds.
Moreover, if $p-1$ and $q-1$ only have large prime factors, then the probability that $p-1$ and $q-1$ share a prime factors is relatively small, and condition $(c)$ is thus likely to hold.
Finally, we note that recently factored RSA numbers\footnote{\url{https://en.wikipedia.org/wiki/RSA_numbers}}, which is a factoring challenge of a set of cryptographically strong moduli organized by the inventors of the RSA cryptosystem, all satisfy both conditions $(b)$ and $(c)$.
For instance, all RSA numbers that have been factored over the last five years (i.e., \texttt{RSA-250}, \texttt{RSA-240}, \texttt{RSA-768}, \texttt{RSA-232} and \texttt{RSA-230}) all have $c' \leq 2$ and $c \leq 8$ in conditions $(b)$ and $(c)$.

\section{Proof of Theorem~\ref{thm:elgamal}}
\label{appendix:elgamal}

\elgamal*

\begin{proof}

For quantum learnability, we first note that a quantum learner can with high probability recover the modulus $N$ from $\mathcal{O}\left(\mathrm{poly}(n)\right)$ examples drawn from the example oracle $E(c_{(N,m)}, \mathcal{D}^U_n)$.
Next, we note that if $N$ is known, then we can use Shor's algorithm~\cite{shor:factoring} to efficiently compute the integer $d \in \{0, \dots, (p-1)(q-1)\}$ such that\begin{align}
    (m^3)^d \equiv m \mod N, \quad  \text{ for all }m \in \mathbb{Z}^*_N.
\end{align}
Next, we note that from examples of the form $(x, \mathrm{bin}(m, k))$ we can retrieve the $k$th bit of $m^3$, where $k = \mathrm{int}(x_1 : \dots : x_{\lfloor \log n \rfloor})$.
Since for any given $k \in [n]$ we have 
\begin{align}
    \mathsf{Pr}_{x \sim \mathcal{D}^U_n}\left(\mathrm{int}(x_1 :\dots:x_{\lfloor\log n \rfloor}) = k \right) = \frac{1}{n}
\end{align}
we find that $\mathcal{O}(\mathrm{poly}(n))$ examples of the form $(x, \mathrm{bin}(m, k))$ suffices to reconstruct the full binary representation of $m^3$ with high probability.
Moreover, we can with high probability obtain an example of the form $(x, \mathrm{bin}(m, k))$ using just $\mathcal{O}(1)$ queries to the example oracle.
Finally, using $d$ and $m^3$ we can compute $(m^3)^d \equiv m \mod N$.

To show classical non-learnability we show that an efficient classical learning algorithm $\mathcal{A}_{\mathrm{learn}}$ can efficiently solve the discrete cube root problem.
To do so, we let $e \in \mathbb{Z}^*_N$ and our goal is to use $\mathcal{A}_{\mathrm{learn}}$ to efficiently compute $m \in \mathbb{Z}^*_N$ such that $m^3 \equiv e\mod N$.
First, we generate examples 
\begin{align}
    (x, \mathrm{bin}(e, k))\text{ and }(x, \mathrm{bin}(N, k))
\end{align}
where $x \in \{0,1\}^n$ is sampled uniformly at random and $k = \mathrm{int}(x_1 : \dots : x_{\lfloor \log n \rfloor})$. 
If we plug these examples into $\mathcal{A}_{\mathrm{learn}}$ with $\epsilon = 1/n^3$ and $\delta = 1/3$, then with high probability we obtain some $m'$ such that
\begin{align}
\label{eq:prob_con}
\mathsf{Pr}_{k \sim [n]}\left(\mathrm{bin}(m^3, k) \neq \mathrm{bin}((m')^3, k)\right) \leq \frac{1}{2n^3},
\end{align}
where $k \in [n]$ is sampled uniformly at random.
Next, we claim that $m^3 \equiv (m')^3 \mod N$.
Specifically, suppose there exists some $i$ such that $\mathrm{bin}(m^3, i) \neq \mathrm{bin}((m')^3, i)$, then this implies that
\begin{align}
\mathsf{Pr}_{k \sim [n]}\left(\mathrm{bin}(m^3, k) \neq \mathrm{bin}((m')^3, k)\right) = \frac{1}{n}\sum_{k = 1}^n\mathds{1}\left[\mathrm{bin}(m^3, k) = \mathrm{bin}((m')^3, k)\right] \geq \frac{1}{n},
\end{align}
which clearly contradicts Eq.~\eqref{eq:prob_con}.
Now since $x \mapsto x^3 \mod N$ is a bijection to and from $\mathbb{Z}_N^*$ we conclude that $m = m'$ and that we have thus solved our instances of the discrete cube root.
\end{proof}

\section{Proof of Theorem~\ref{thm:seps_no-eff-data}}
\label{appendix:seps_no-eff-data}

\noeffdata*

\begin{proof}

Firstly, a quantum learner can iterate over all concepts in $\mathcal{C}_n$ and find the one that matches the examples obtained from the oracle.
In other words, a quantum learner can implement empirical risk minimization through brute-force search.
By Corollary 2.3 of~\cite{shalev:book} this shows that $L \in \mathsf{QQ}$.

Next, suppose $L \in \mathsf{CC}$, i.e., suppose there exists an efficient classical learning algorithm for $L$ that uses a classically evaluatable hypothesis class.
By combining the classical learning algorithm with the evaluation algorithm of the hypothesis class we obtain a polynomial-time classical randomized algorithm $\mathcal{A}$ such that for every $c'_n \in \mathcal{C}_n$ on input $\mathcal{T} = \{(x_i, c'_n(x_i)) \mid x_i \sim \mathcal{D}_n\}_{i = 1}^{\mathrm{poly}(n)}$ and $x \in \{0,1\}^n$ we have
\[
    \mathsf{Pr}_{x \sim \mathcal{D}_n}\left[\mathsf{Pr}\left(\mathcal{A}(x, 0^{\lfloor 1/\epsilon \rfloor}, \mathcal{T}) = c'_n(x)\right) \geq \frac{2}{3} \right] \geq 1- \epsilon
\]
If we apply $\mathcal{A}$ to the concepts $\{c_n\}_{n \in \mathbb{N}}$ we obtain $(\{c_n\}_{n \in \mathbb{N}}, \{\mathcal{D}_n\}_{n\in \mathbb{N}}) \in \mathsf{HeurBPP/samp} \subseteq\mathsf{HeurP/poly}$, which contradicts the classical non-learnability assumption.
Therefore, it must hold that $L \not\in \mathsf{CC}$.

\end{proof}

\subsection{Proof of Lemma~\ref{lemma:1}}
\label{appendix:lemma1}

\lemmaone*

\begin{proof}

Let $L' \in \mathsf{BQP}$-$\mathsf{complete}$ and consider the many-to-one polynomial-time reduction $f: L \rightarrow L'$ such that $L(x) = L'(f(x))$.
Also, consider the pushforward distributions $\mathcal{D}'_n = f(\mathcal{D}_n)$ on $\{0,1\}^{n'}$\footnote{Note that $f$ can map instances $x \in \{0,1\}^n$ to instances of size $n'$ that are at most polynomially larger than $n$.}, i..e, the distribution induced by first sampling $x \sim \mathcal{D}_n$ and subsequently computing $f(x)$.
Next, we suppose that $(L', \mathcal{D}') \in \mathsf{HeurP/poly}$.
Specifically, we suppose that there exists a classical algorithm $\mathcal{A}$ and a sequence advice strings $\{\alpha_n\}_{n\in \mathbb{N}}$ as in Definition~\ref{def:heurp/poly} such that for every $n \in \mathbb{N}$:
\begin{align}
   \mathsf{Pr}_{y \sim \mathcal{D}'_n}\left[ \mathcal{A}(y, 0^{\lfloor 1/\epsilon \rfloor}, \alpha_{n'}) = L'(y)\right] \geq 1-\epsilon 
\end{align}
By the definition of the push-forward distribution $\mathcal{D}'_n$ we have
\begin{align}
\label{eq:proof1_lemma1}
   \mathsf{Pr}_{y \sim \mathcal{D}'_n}\left[ \mathcal{A}(y, 0^{\lfloor 1/\epsilon \rfloor}, \alpha_{n'}) = L'(y)\right] = \mathsf{Pr}_{x \sim \mathcal{D}_n}\left[ \mathcal{A}(f(x), 0^{\lfloor 1/\epsilon \rfloor}, \alpha_{n'}) = L'(f(x))\right]
\end{align}
Finally, we define a polynomial-time classical algorithm $\mathcal{A}'$ that uses advice as follows
\[
\mathcal{A}'(x, 0^{\lfloor 1/\epsilon \rfloor}, \alpha_n) = \mathcal{A}(f(x), 0^{\lfloor 1/\epsilon \rfloor}, \alpha_n).
\]
Then, by Eq.~\eqref{eq:proof1_lemma1} we have that
\begin{align}
   \mathsf{Pr}_{x \sim \mathcal{D}_n}\left[ \mathcal{A}'(x, 0^{\lfloor 1/\epsilon \rfloor}, \alpha_n) = L(x)\right] = \mathsf{Pr}_{x \sim \mathcal{D}_n}\left[\mathcal{A}(f(x), 0^{\lfloor 1/\epsilon \rfloor}, \alpha_n) = L'(f(x)) \right]\geq 1-\epsilon 
\end{align}
which implies that $(L, \mathcal{D}) \in \mathsf{HeurP/poly}$.
This contradicts the assumptions made in the lemma, and we therefore conclude that indeed $(L', \mathcal{D}') \not\in \mathsf{HeurP/poly}$.

\end{proof}

\subsection{Proof of Lemma~\ref{lemma:2}}
\label{appendix:lemma2}

\lemmatwo*

\begin{proof}

Since $L$ is polynomially random self-reducible (for a formal definition we refer to~\cite{feigenbaum:rsr}), we know that there exists a family of distributions $\mathcal{D} = \{\mathcal{D}_n\}_{n \in \mathbb{N}}$, a polynomial-time computable function $f$, and some integer $k_n = \mathcal{O}(\mathrm{poly}(n))$ such that 
\begin{align}
\label{eq:proof1_lemma2}
    \mathsf{Pr}_{y_1, \dots, y_{k_n} \sim \mathcal{D}_n}\Big(f\big(x, L(y_1), \dots, L(y_{k_n})\big) = L(x)\Big) \geq \frac{3}{4}
\end{align}

Suppose $(L, \mathcal{D}) \in \mathsf{HeurP/poly}$, i.e., there exists a polynomial-time classical algorithm $\mathcal{A}$ and a sequence of advice strings $\{\alpha_n\}_{n \in \mathbb{N}}$ such that
\begin{align}
\label{eq:proof2_lemma2}
    \mathsf{Pr}_{y \sim \mathcal{D}_n}\left(\mathcal{A}(y, 0^{\lfloor 1/\epsilon \rfloor}, \alpha_n) = L(y) \right) \geq 1 - \epsilon.
\end{align}
Let $\epsilon' = 1/(9k)$, then by combining Eq.~\eqref{eq:proof1_lemma2} and Eq.~\eqref{eq:proof2_lemma2} we get that
\begin{align}
\mathsf{Pr}_{y_1, \dots, y_{k_n} \sim \mathcal{D}_n(x)}\Big(f\big(x, \mathcal{A}(y_1, 0^{1/\epsilon'}, \alpha_n), \dots, \mathcal{A}(y_{k_n}, 0^{1/\epsilon'}, 
\alpha_n)\big) = L(x)\Big) \geq \frac{2}{3}
\end{align}
In other words, if we define $\mathcal{A}'(x, \alpha_n) = f\big(x, \mathcal{A}(y_1, 0^{1/\epsilon'}, \alpha_n), \dots, \mathcal{A}(y_1, 0^{1/\epsilon'}, 
\alpha_n)\big)$, where $y_i\sim \mathcal{D}_n$ are sampled during the runtime of the algorithm, then we conclude that $L \in \mathsf{BPP/poly} = \mathsf{P/poly}$.
This contradicts the assumption in the lemma, and we therefore conclude that $(L, \mathcal{D}) \not\in \mathsf{HeurP/poly}$.
    
\end{proof}

\subsection{Proof of Corollary~\ref{cor:seps1}}
\label{appendix:cor2}

\cortwo*

\begin{proof}

    Suppose there exists $L\in \mathsf{BQP}$-$\mathsf{complete}$ such that for every efficiently samplable distribution $\mathcal{D}$ we have $(L, \mathcal{D}) \in \mathsf{HeurBPP}/\mathsf{samp}$.
    Now let $f$ be the (many-to-one) reduction from $\mathsf{DLP}$ to $L$ and define $
    \mathcal{D}_L = f(\mathcal{D}_U)$ to be the push-forward distribution of the uniform distribution $\mathcal{D}_U$ under $f$.
    By our assumption at the beginning of the proof, we know $(L, \mathcal{D}_L) \in \mathsf{HeurBPP}/\mathsf{samp}$. 
    However, we now note that the $\mathsf{HeurBPP}/\mathsf{samp}$-algorithm for solving $(L, \mathcal{D}_L)$ actually also solves $(\mathsf{DLP}, \mathcal{D}_U)$, which implies that $(\mathsf{DLP}, \mathcal{D}_U) \in \mathsf{HeurBPP}/\mathsf{samp}$. 
    Since $\mathsf{DLP}$ is both random self-reducible as well as random verifiable, we conclude that this implies that $\mathsf{DLP} \in \mathsf{BPP}$.
\end{proof}

\section{Proof of Theorem~\ref{thm:limitations_huang}}
\label{appendix:limitations_huang}

\limitationshuang*

\begin{proof}

We define $\mathsf{DLP}$ to be the problem of computing the first bit of $\log_a x$ (i.e., the smallest positive integer $\ell$ such that $a^\ell \equiv x \mod p$)  with respect to a generator $a \in \mathbb{Z}^*_p$ for a given $x \in \mathbb{Z}^*_p$.

First, using Shor's algorithm we can construct a polynomial-depth circuit $U_{\mathrm{Shor}}$ such that
\begin{align}
    U_{\mathrm{Shor}}\ket{0, x, 0^\ell} = (1-\alpha)\ket{\mathsf{DLP}(x), x, 0^\ell} + \alpha\ket{\mathrm{garbage}},
\end{align}
for all $x \in \{0, 1\}^n$ and where $\alpha = \mathcal{O}(2^{-n})$.
Next, we parameterize 
\begin{align}
U(x) = U \cdot \left( I_0 \otimes \left[\bigotimes^n_{i = 1}X_i(g_\gamma(x_i) \cdot 2\pi)\right]\right),
\end{align}
where $X$ is a rotation such that $X(0)\ket{0} = \ket{0}$ and $X(2\pi)\ket{0} = \ket{1}$, and $g_\gamma$ is a continuous function such that
\begin{align}
    g_\gamma(x_i) = \begin{cases} 0, & x_i \in [-1, -\gamma)\\ (2\gamma - x)(\gamma + x)/(4\gamma^3), & x\in (-\gamma, \gamma)\\ 1, & x_i \in (\gamma, 1] \end{cases},
\end{align}
for some $\gamma > 0$.
Finally, we add $2T$ layers of identities to $U(x)$, where $T$ denotes the depth of $U_\mathrm{Shor}$.

We define $H(x)$ to be the Hamiltonian family on $\mathbb{C}^{2^{s}} \oplus \mathbb{C}^{2^{3T}}$ with $s = n + \ell + 1$  given by
\begin{align}
\label{eq:kitaev}
H(x) = H_{\mathrm{init}} + H_{\mathrm{clock}} + \sum_{t = 1}^{3T} H_t(x),
\end{align}
with
\begin{align}
H_{\mathrm{init}} &=  \sum_{i = 1}^s \ket{0}\bra{0}_i,\\
H_{\mathrm{clock}} &=  \sum_{t = 1}^{3T - 1}\ket{01}\bra{01}^{\mathrm{clock}}_{t, t+1},\\
H_t(x) &= \frac{1}{2}\Big(I \otimes \ket{100}\bra{100}^{\mathrm{clock}}_{t-1, t, t+1} + I \otimes \ket{110}\bra{110}^{\mathrm{clock}}_{t-1, t, t+1} -\\
&U_t(x) \otimes \ket{110}\bra{100}^{\mathrm{clock}}_{t-1, t, t+1} - U_t(x)^\dagger \otimes \ket{100}\bra{110}^{\mathrm{clock}}_{t-1, t, t+1}\Big)
\end{align}
where $\ket{.}\bra{.}_i$ acts on the $i$th site of $\mathbb{C}^{2^s}$, $\ket{.}\bra{.}_j^{\mathrm{clock}}$ acts on the $j$th site of $\mathbb{C}^{2^{3T}}$ and $U_t$ denotes the $t$th layer of gates in $U(x)$. 
Note that $H(x)$ is 5-local for all $x \in [-1, 1]$.
The ground state of $H(x)$ is given by $\rho(x) = \ket{\psi(x)}\bra{\psi(x)}$, where
\begin{align}
\ket{\psi(x)} = \frac{1}{\sqrt{3T}} \sum_{t = 1}^T(U_t\cdots U_1)(x)\ket{0^s}\ket{1^t0^{3T - t}},
\end{align}

We define $O = \ket{0}\bra{0}_0 \otimes I \otimes \ket{1}\bra{1}_T^{\mathrm{clock}}$ and note that it is a local observable with constant norm.
Now $f_{H, O}$ defined in Eq.~\ref{eq:huang_science} is such that 
\begin{align}
f_{H,O}(x) &= \mathrm{Tr}\left[\rho_0(x)O \right] = \frac{2}{3} p_1(x),
\end{align}
where $p_1(x)$ denotes the probability that $\ket{\psi_{\mathrm{out}}(x)} = U(x)\ket{0^s}$ outputs 1 when measuring the first qubit in the computational basis.
In particular, we have that
\begin{align}
f_{H,O}(x) &= \mathrm{Tr}\left[\rho_0(x)O \right] = \frac{2}{3} \left( (1-\alpha)\mathsf{DLP}(g_\gamma(x)) + \alpha \cdot \mathrm{garb}\right),
\end{align}
for all $x \in [-1, 1]^n$ for which there does not exists $x_i \in (-\gamma, \gamma)$, and some quantity $\mathrm{garb} \leq 1$.

Finally, assume that we obtain $\overline{f}_{H, O}$ such that
\begin{align}
\label{eq:guarnatee}
\mathbb{E}_{x \sim [-1, 1]^n}\Big[ \left|\overline{f}_{H, O} - f_{H,O} \right| \Big] < \frac{1}{6}.
\end{align}
Also, suppose there exists a bitstring $y \in \{0,1\}^n$ whose corresponding corner 
$C_y\subset [-1,1]^n$\footnote{Here $y \in \{0,1\}^n$ is mapped to $\{-1,1\}^n$ by setting all $0$s to $-1$, and the corner $C_y$ consists of all points $x \in [-1, 1]^n$ whose $i$th coordinate is $\gamma$ close to $y_i$ for all $i\in [n]$.} with size $\gamma$ is such that
\begin{align}
    \left|\mathbb{E}_{x \sim C_y}\Big[ \overline{f}_{H, O} \Big] - \mathbb{E}_{x \sim C_y}\Big[ f_{H,O}\Big] \right| > \frac{1}{3}.
\end{align}
Then, we find that
\begin{align}
    \mathbb{E}_{x \sim [-1, 1]^n}\Big[ \left|\overline{f}_{H, O} - f_{H,O} \right| \Big] &= \int_{[-1, 1]^n} \left|\overline{f}_{H, O} - f_{H,O} \right| dx\\
    &\geq \int_{C_y} \left|\overline{f}_{H, O} - f_{H,O} \right| dx\\
    &\geq \left|\int_{C_y} \overline{f}_{H, O} - f_{H,O}  dx\right|\\
    &=\left|\left(\int_{C_y} \overline{f}_{H, O}dx\right) - \left(\int_{C_y} f_{H, O}dx\right)\right|\\
    &= \left|\mathbb{E}_{x \sim C_y}\Big[ \overline{f}_{H, O} \Big] - \mathbb{E}_{x \sim C_y}\Big[ f_{H,O}\Big] \right| > \frac{1}{3}.
\end{align}
which clearly contradicts Eq.~\eqref{eq:guarnatee}.
We therefore conclude that 
\begin{align}
    \left|\mathbb{E}_{x \sim C_y}\Big[ \overline{f}_{H, O} \Big] - \mathbb{E}_{x \sim C_y}\Big[ f_{H,O}\Big] \right| < \frac{1}{3}.
\end{align}

In conclusion, for every $y \in \{0,1\}^n$ the quantity $\mathbb{E}_{x \sim C_y}[\overline{f}_{H, O}]$ is exponentially close to $\mathsf{DLP}(y)$.
Finally, we can efficiently estimate $\mathbb{E}_{x \sim C_y}[\overline{f}_{H, O}]$ to within additive inverse-polynomial error, which allows us to compute $\mathsf{DLP}(y)$ in $\mathsf{BPP/poly} = \mathsf{P/poly}$.

\end{proof}

\paragraph{Spectral gap and smoothness}
Note that the Hamiltonian family constructed above indeed does not satisfy all requirements for the methods of Huang et al.~\cite{huang:science} to work.
In particular, it is known that the spectral gap of Hamiltonians obtained from Kitaev's circuit-to-Hamiltonian construction (i.e., those defined in Eq.~\eqref{eq:kitaev}) have a spectral gap that is inverse polynomial in the depth of the circuit, which is our case is polynomial in the instance size $n$.
Moreover, since we apply a function $g_\gamma$ to the parameters $x$ (which has a rapid increase between $-\gamma$ and $\gamma$), it is likely that the average gradient of the function $\mathrm{Tr}\left[O\rho_H(x)\right]$ is not bounded by a constant, but rather scales with the number of parameters $m$ (which in our case also scales with the instance size $n$).

\end{document}